\documentclass[bib]{ba}

\usepackage{graphicx,hyperref}
\usepackage{bayes}
\usepackage{amsmath}
\usepackage{amssymb}
\usepackage{float}
\usepackage{caption}
\usepackage[figuresright]{rotating}

\newtheorem{theorem}{Theorem}[section]
\newtheorem{lemma}[theorem]{Lemma}

\newtheorem{corollary}[theorem]{Corollary}
\newtheorem{remark}[theorem]{Remark}

\newenvironment{proof}[1][Proof]{\begin{trivlist}
\item[\hskip \labelsep {\bfseries #1}]}{\end{trivlist}}

\newcommand{\qed}{\nobreak \ifvmode \relax \else
      \ifdim\lastskip<1.5em \hskip-\lastskip
      \hskip1.5em plus0em minus0.5em \fi \nobreak
      \vrule height0.75em width0.5em depth0.25em\fi}

\begin{document}

\inserttype[ba0001]{article}
\renewcommand{\thefootnote}{\fnsymbol{footnote}}
\author{F. Ruggeri, Z. Sawlan, M. Scavino and R. Tempone }{
 \fnms{Fabrizio}
 \snm{Ruggeri} \!\!\!
 \footnotemark[1]\ead{fabrizio@mi.imati.cnr.it}\!\!
,
  \fnms{Zaid}
  \snm{Sawlan}\!
  \footnotemark[2]\ead{zaid.sawlan@kaust.edu.sa}\!\!
,
  \fnms{Marco}
  \snm{Scavino}\!
  \footnotemark[3]\ead{marco.scavino@kaust.edu.sa}
and
  \fnms{Raul}
  \snm{Tempone}\!
  \footnotemark[4]\ead{raul.tempone@kaust.edu.sa}
}

\title[Bayesian Inference for Linear Parabolic PDEs]{A hierarchical Bayesian setting for an inverse problem in linear parabolic PDEs with noisy boundary conditions}

\maketitle

\footnotetext[1]{
 CNR - IMATI, Consiglio Nazionale delle Ricerche, Milano, Italy
 \href{mailto:email1@example.com}{fabrizio@mi.imati.cnr.it}
}

\footnotetext[2]{
Corresponding author. CEMSE, King Abdullah University of Science and
Technology, Thuwal, 23955-6900, KSA
 \href{mailto:email1@example.com}{zaid.sawlan@kaust.edu.sa}
}

\footnotetext[3]{
CEMSE, King Abdullah University of Science and
Technology, Thuwal, 23955-6900, KSA
 \href{mailto:email1@example.com}{marco.scavino@kaust.edu.sa} , and Instituto de Estad\'{\i}stica (IESTA), Universidad de la Rep\'ublica, Montevideo, Uruguay
  \href{mailto:email2@example.com}{mscavino@iesta.edu.uy}
}

\footnotetext[4]{
CEMSE, King Abdullah University of Science and
Technology, Thuwal, 23955-6900, KSA
 \href{mailto:email1@example.com}{raul.tempone@kaust.edu.sa}
}

\renewcommand{\thefootnote}{\arabic{footnote}}

\begin{abstract}
In this work we develop a Bayesian setting to infer unknown parameters in initial-boundary value problems related to linear parabolic partial differential equations. We realistically assume that the boundary data are noisy, for a given prescribed initial condition. We show how to derive the joint likelihood function for the forward problem, given some measurements of the solution field subject to Gaussian noise. Given Gaussian priors for the time-dependent Dirichlet boundary values, we analytically marginalize the joint likelihood using the linearity of the equation. Our hierarchical Bayesian approach is fully implemented in an example that involves the heat equation. In this example, the thermal diffusivity is the unknown parameter. We assume that the thermal diffusivity parameter can be modeled a priori through a lognormal random variable or by means of a space-dependent stationary lognormal random field. Synthetic data are used to test the inference. We exploit the behavior of the non-normalized log posterior distribution of the thermal diffusivity. Then, we use the Laplace method to obtain an approximated Gaussian posterior and therefore avoid costly Markov Chain Monte Carlo computations. Expected information gains and predictive posterior densities for observable quantities are numerically estimated using Laplace approximation for different experimental setups.

\keywords{\kwd{Linear Parabolic PDEs},  \kwd{Noisy Boundary Parameters}, \kwd{Bayesian Inference}, \kwd{Heat Equation}, \kwd{Thermal Diffusivity}.}
\end{abstract}

\section{Introduction}

Parabolic partial differential equations model various important physical phenomena such as diffusion and heat transport. The solution of such equations propagates forward in time from an initial condition given boundary conditions and equation coefficients. In applications, some equation coefficients can be unknown quantities that need to be estimated. In addition, exact initial and boundary conditions might not be known. One possible approach to estimate these unknowns is to solve the inverse problem given some information about the solution. Classical inversion methods for parabolic partial differential equations are introduced by \cite{Samarskii} in the Chapter 8 of their book. 

In this work, we consider a Bayesian inversion problem to determine the coefficients of linear parabolic partial differential equations, under the assumption that noisy measurements are available in the interior of a domain of interest and for the unknown boundary conditions. A main novelty of our approach to solve the inverse problem relies on the assumption that the boundary parameters are unknown and modeled by means of adequate probability distributions. Subsequently, the contribution of the boundary parameters is marginalized out from the joint law with the unknown equation coefficients we want to infer, allowing the characterization of their posterior distribution. There are many advantages to a Bayesian approach. For example, it provides a solution along with a comprehensive measure of uncertainty given by the posterior distribution. Moreover, the prior available information can be easily incorporated in terms of elicited prior distributions \citep{Ghosh}. An important issue in this work is that Bayesian inversion is posed as a hierarchical process. Boundary conditions can therefore be treated as unknown parameters \citep{Kaipio}. Since we are only interested in estimating the equation coefficients, we eliminate those extra parameters by marginalization.

Bayesian inversion techniques for the heat equation have been discussed and implemented in some previous works. \cite{Kaipio} provided a general Bayesian framework for inverse problems in heat transfer, classified according to the dominant mode in heat transfer. The authors address many issues regarding forward problems and their statistical analysis. The prior modeling is extensively discussed, as well as how to deal, in particular, with different sources of uncertainties. Heat flux reconstruction problems have been studied by \cite{Wang1, Wang2}. When referring to the problem of parameter estimation in inverse heat conduction problems, Wang and Zabaras showed how to infer the thermal conductivity using a hierarchical Bayesian framework, on the basis of temperature readings within a conducting solid, assuming that the heat flux on the boundary and the heat source are known. They also explored the high dimensional posterior state space by means of Markov Chain Monte Carlo simulation. \cite{Lanz} estimated the thermal conductivity of a polymer, transforming the heat equation into a stochastic differential equation and considering the Euler-Maruyama approximation to get the likelihood, introducing latent observations in space and then using a relatively cumbersome Markov Chain Monte Carlo method. \cite{Fudym1} and \cite{Fudym2} addressed the estimation problem of the thermal diffusivity, as in the present work, which is a parameter that describes thermophysical property of materials. In their works a large number of temperature measurements is made by an infrared camera, with fine spatial resolution and high frequency. They solved one and two-dimensional forward problems for transient heat conduction, with spatially varying thermal conductivity and volumetric heat capacity, by finite differences, according to a nodal strategy. The parameter vector at each node is then estimated either by minimizing an a posteriori objective function when prior Gaussian distributions are assumed for the parameters, or by means of Markov Chain Monte Carlo methods for different prior distributions. 

This work is organized as follows.
In Section \ref{sec2}, we introduce the statistical setting and we derive the explicit form of the joint law of the unknown equation coefficients and the boundary parameters.
In Section \ref{sec3}, we use a finite element scheme in order to write the solution of the forward problem as a linear function of the boundary conditions. We demonstrate in Section \ref{sec4} that, under certain conditions, an exact marginalization can be carried out, yielding a closed formula for the marginal likelihood of the equation coefficients.
In Section \ref{sec5}, we apply our marginalization technique to  estimating thermal diffusivity in the one-dimensional heat equation in two cases given temperature simulated data. Numerical results are obtained for the non-normalized log posterior distribution of the thermal diffusivity. We model prior knowledge about the thermal diffusivity first as a lognormal random variable and then using a lognormal random field with a squared exponential (SE) covariance function \citep{Rasmussen}. In the first case, we use the Laplace method to provide an approximated Gaussian posterior distribution for the thermal diffusivity. Such method is then applied to obtain fast estimations of the information gain and the expected information gain under three experimental setups, and the predictive posterior mean of the temperature is also derived using the inferred thermal diffusivity. 
In the second case, where the thermal diffusivity is allowed to depend on the spatial variable, the Laplace approximation is used to obtain the posterior distribution of the hyperparameters that characterize the prior distribution for the thermal diffusivity.   

\section{Statistical setting and preliminary results}
\label{sec2}

In this section we introduce the statistical model associated to the 
forward initial-boundary value problems for linear parabolic partial differential equations. We then derive, under mild assumptions, the exact expression for the joint likelihood function of the unknown parameters in the parabolic equation and the unknown boundary parameters. \\

Consider the deterministic one-dimensional parabolic initial-boundary value problem:
\begin{equation}
\begin{cases}
\partial_t T + L_{\boldsymbol{\theta}}T = 0, & x \in (x_L, x_R), \, 0 < t \leqslant t_N < \infty \\
T(x_L,t) = T_L(t), & t \in [0,t_N] \\
T(x_R,t) = T_R(t), & t \in [0,t_N] \\
T(x,0) = g(x), & x \in(x_L, x_R)\,,
\end{cases}
\label{eq:main1}
\end{equation}

\noindent where $L_{\boldsymbol{\theta}}$ is a linear second-order partial differential operator that takes the form
\[ L_{\boldsymbol{\theta}}T = - \partial_x( a(x) \partial_x T) + b(x)\partial_x T + c(x)T,  \]
$\boldsymbol{\theta}(x) = (a(x),b(x),c(x))^{tr},$ and the partial differential operator $\partial_t + L_{\boldsymbol{\theta}}$ is  parabolic, because (\cite{Evans:1998}, [p.372]) there exists $\epsilon$ such that $a(x) \geqslant \epsilon >0$ for all $x \in (x_L, x_R)$. We also assume that
\begin{enumerate}
\item[P1] $a, b$ and $c$ are bounded functions.
\item[P2] $T_L, T_R$ and $g$ are square integrable functions.
\item[P3] The initial condition, $g$, is consistent with the boundary functions, namely $g(x_L) = T_L(0)$ and $g(x_R) = T_R(0)$.
\end{enumerate}
Then, under the assumptions P1-P3, there exists a unique weak solution of  \eqref{eq:main1} (\cite{Evans:1998}, [pp.375-377]).

Our main objective is to provide a Bayesian solution to an inverse problem for $\boldsymbol{\theta}$, where we assume that
\begin{enumerate}
\item[i]  $\boldsymbol{\theta}$ is unknown, while the initial condition $g$ in the initial-boundary value problem is known;
\item[ii] $\boldsymbol{\theta}$ is allowed to vary with the spatial variable $x$.
\end{enumerate}

\begin{remark}
In our Bayesian approach, we will assume later that the coefficient $a(x)$ is a lognormal random variable or lognormal random field. Therefore, $a(x)$ will not be bounded as assumed in P1.
However, it can be proved that there exists a unique solution of the stochastic parabolic initial-boundary value problem in the space $L^2(\Omega, H^{1})$. Such proof can be found in \cite{Charrier:2012} for elliptic boundary value problems but it can be also extended to parabolic initial-boundary value problems.
\end{remark}

Given noisy readings of the function $T(x,t)$ at the $I+1$ spatial locations, including the boundaries, $x_L=x_0, x_1, x_2, \ldots, x_{I-1}, x_I = x_R$,
at each of the $N$ times $t_1, t_2, \ldots, t_N$,
we want to infer $\boldsymbol{\theta}$ using a Bayesian approach. To determine the posterior distribution for $\boldsymbol{\theta}$, we need first to obtain the likelihood function of $\boldsymbol{\theta}$.
The remainder of this section derives the joint likelihood function of $\boldsymbol{\theta}$ and the boundary parameters. Let us introduce some convenient notation and assumptions: let $ \mathbf{Y_n} := (Y_{0,n}, \ldots, Y_{I,n})^{tr}$ denote the vector of observed readings at time $t_n$, and assume a statistical model with an additive Gaussian experimental noise $\boldsymbol{\epsilon_n}$; that is:
\begin{equation}
\mathop{\mathbf{Y_n}}\limits_{(I+1) \times 1} = \left[
\begin{array}{c}
T_{L}(t_n) \\
T(x_1,t_n) \\
\vdots \\
T(x_{I-1},t_n) \\
T_{R}(t_n)
\end{array} \right] + \boldsymbol{\epsilon_n}, \label{eq:addmodel}
\end{equation}

\noindent where $\boldsymbol{\epsilon_n} \stackrel{\scriptsize{\textrm{i.i.d.}}}{\sim} {\mathcal{N}}(\mathbf{0}_{I+1}, \sigma^2 \,\mathbf{I}_{I+1})$ for some measurement error variance $\sigma^2 > 0$. The covariance matrix of $\boldsymbol{\epsilon_n}$ is assumed equal to $\sigma^2 \,\mathbf{I}_{I+1}$ for simplicity, a general covariance matrix $\Sigma_{\boldsymbol{\epsilon_n}}$ could be used as well provided that the boundary measurement errors are independent from the interior measurement errors.

Also denote by $\mathop{\mathbf{Y_n^{I}}}\limits_{(I-1) \times 1} := (Y_{1,n}, \ldots, Y_{I-1,n})^{tr}$ the vector of observed data at the interior locations
$x_1, x_2, \ldots, x_{I-1}$ and let $\mathbf{Y_n^{B}} := (Y_{L,n}, Y_{R,n})^{tr}$ be the vector of observed data at the boundary locations $x_0,x_I$ at time $t_n.$\\
The density of $\mathbf{Y_n^{I}}$ is derived as it follows. First consider the time local problem, defined between consecutive measurement times, i.e.
\begin{equation}
\begin{cases}
\partial_t T + L_{\boldsymbol{\theta}} T = 0, & x \in (x_L, x_R),\: t_{n-1} < t \leqslant t_n,\\
T(x_L, t) =  T_L(t), & t \in [t_{n-1},t_n],\\
T(x_R, t) =  T_R(t), & t \in [t_{n-1},t_n],\\
T(x, t_{n-1}) = \widehat{T}(x, t_{n-1}), & x \in (x_L, x_R)\,,
\label{locprob}
\end{cases}
\end{equation}
whose exact solution, denoted by $\widehat{T}(\cdot,t_n)$, depends only on $\boldsymbol\theta, \widehat{T}(\cdot, t_{n-1})$ and the boundary values $\left\{ T_L(t), T_R(t) \right\}_{t \in (t_{n-1}, t_n)}$. Finally, use the form of the statistical model (\ref{eq:addmodel}) to obtain the form of the density of $\mathbf{Y_n^{I}}$.

\begin{lemma}
\label{lemma1}
Given the model (\ref{eq:addmodel}), the probability density function of $\mathbf{Y_n^{I}}$ is expressed as
\begin{equation}
\rho(\mathbf{Y_n^{I}} | \theta, \widehat{T}(\cdot, t_{n-1}), \left\{ T_L(t), T_R(t) \right\}_{t \in (t_{n-1}, t_n)} )
= \frac{1}{(\sqrt{2 \pi} \sigma)^{I-1}} \exp \left( - \frac{1}{2 \sigma^2} \left\| {\mathbf{R}}_{t_{n}}\right\|^{2}_{\ell^2} \right),
\label{eq:dens}
\end{equation}
where
$\mathop{{\mathbf{R}}_{t_n}}\limits_{(I-1) \times 1} := (\widehat{T}(x_1,t_n) -  Y_{1,n},\, \ldots,\, \widehat{T}(x_{I-1},t_n) - Y_{I-1,n})^{tr}$ denotes the data residual vector at time $t=t_n$.
\end{lemma}

For illustration purposes and without loss of generality, assume now that the Dirichlet boundary condition functions, $T_L(\cdot)$ and $T_R(\cdot)$, are well approximated by piecewise linear continuous functions in the time  partition $\left\{t_n\right\}_{n=1,\ldots,N}$.\\
In this way, only $2N$ parameters, say $T_L(t_n)=T_{L,n}, \,T_R(t_n)=T_{R,n}\,, \:n=1,2, \ldots, N\,,$ suffice to determine the boundary conditions that are, in principle, infinite dimensional parameters. Let $LBC_n$ denote the time nodes that determine the local boundary conditions $\left\{ T_{L,n-1}, T_{L,n}, T_{R,n-1}, T_{R,n} \right\}\,.$ Other interpolation schemes may be used as well.

\begin{remark}
Given the discretized Dirichlet boundary conditions introduced above, we can say that $\widehat{T}(\cdot,t_n)$ depends only on $\boldsymbol\theta, T(\cdot, t_{n-1})$ and the boundary parameters $LBC_n$. Similarly, $\widehat{T}(\cdot,t_{n-1})$ depends on $\boldsymbol\theta, T(\cdot, t_{n-2})$ and the boundary parameters $LBC_{n-1}$. From this recursion, we can obtain 
\begin{equation}
\label{eq:rec}
\rho(\mathbf{Y_n^{I}} | \theta, g, \left\{LBC_j \right\}_{j=1,\ldots,n} ) = \rho(\mathbf{Y_n^{I}} |  \theta, \widehat{T}(\cdot, t_{n-1}), LBC_n ).
\end{equation}
Since the initial condition, $g$, is assumed to be known, it will be omitted in the rest of the paper. 
\end{remark}

\begin{lemma}
Given the model (\ref{eq:addmodel}) and Lemma \ref{lemma1}, the joint likelihood function of $\boldsymbol{\theta}$ and the boundary parameters $\left\{LBC_n\right\}_{n=1,\ldots,N}$ is given by
\begin{eqnarray}
&\,&\rho(\mathbf{Y_1}, \ldots, \mathbf{Y_N} | \boldsymbol{\theta}, \left\{LBC_n\right\}_{n=1,\ldots,N}) = \prod_{n=1}^{N} \frac{1}{(\sqrt{2 \pi} \sigma)^{I-1}} \exp \left( - \frac{1}{2 \sigma^2} \left\| {\mathbf{R}}_{t_{n}}\right\|^{2}_{\ell^2} \right) \nonumber \\
&\,& \times \frac{1}{\sqrt{2 \pi \sigma^2}}  \exp \left( - \frac{1}{2 \sigma^2} \left( T_{L,n} - Y_{L,n} \right)^2 \right)
\times \frac{1}{\sqrt{2 \pi \sigma^2}}  \exp \left( - \frac{1}{2 \sigma^2} \left( T_{R,n} - Y_{R,n} \right)^2 \right) . \label{eq:globlike}
\end{eqnarray}
\end{lemma}

\begin{proof} 
Observe that $\mathbf{Y_n^{I}}, \mathbf{Y_n^{B}}$ are conditionally independent given
$\boldsymbol{\theta}, \widehat{T}(\cdot, t_{n-1}), LBC_n\,.$ Thus, we have
\begin{eqnarray*}
\rho(\mathbf{Y_n}| \boldsymbol{\theta}, \widehat{T}(\cdot, t_{n-1}), LBC_n) &=&
\rho(\mathbf{Y_n^{I}}, \mathbf{Y_n^{B}} | \boldsymbol{\theta}, \widehat{T}(\cdot, t_{n-1}), LBC_n),\\
&=& \rho(\mathbf{Y_n^{I}} | \boldsymbol{\theta}, \widehat{T}(\cdot, t_{n-1}), LBC_n ) \times \rho(\mathbf{Y_n^{B}} | \boldsymbol{\theta}, \widehat{T}(\cdot, t_{n-1}), LBC_n ),\\
&=& \rho(\mathbf{Y_n^{I}} | \boldsymbol{\theta}, \widehat{T}(\cdot, t_{n-1}), LBC_n) \times \rho(\mathbf{Y_n^{B}} | LBC_n)
\end{eqnarray*}
since $\mathbf{Y_n^{B}} | LBC_n$ does not depend on either $\boldsymbol{\theta}$ nor $\widehat{T}(\cdot, t_{n-1})$.

The joint likelihood function can then be written as
\[ \rho(\mathbf{Y_1}, \ldots, \mathbf{Y_N} | \boldsymbol{\theta}, \left\{LBC_n\right\}_{n=1,\ldots,N})  \]
\[ = \rho(\mathbf{Y_N} | \boldsymbol{\theta}, \left\{LBC_n\right\}_{n=1,\ldots,N}, \mathbf{Y_{N-1}}, \ldots, \mathbf{Y_1}) \times
\rho(\mathbf{Y_1}, \ldots, \mathbf{Y_{N-1}} | \boldsymbol{\theta}, \left\{LBC_n\right\}_{n=1,\ldots,N-1})  \]
(since $\mathbf{Y_N} | \boldsymbol{\theta}, \left\{LBC_n\right\}_{n=1,\ldots,N}$ is independent from $\mathbf{Y_1}, \ldots, \mathbf{Y_{N-1}}$)
\[ = \rho(\mathbf{Y_N} |  \boldsymbol{\theta}, \left\{LBC_n\right\}_{n=1,\ldots,N}) \times
\rho(\mathbf{Y_1}, \ldots, \mathbf{Y_{N-1}} | \boldsymbol{\theta}, \left\{LBC_n\right\}_{n=1,\ldots,N-1})  \]
(using equation \eqref{eq:rec})
\[ = \rho(\mathbf{Y_N} | \boldsymbol{\theta}, \widehat{T}(\cdot, t_{N-1}), LBC_N ) \times
\rho(\mathbf{Y_1}, \ldots, \mathbf{Y_{N-1}} | \boldsymbol{\theta}, \left\{LBC_n\right\}_{n=1,\ldots,N-1})   \]
(iterating the previous arguments)
\[ =  \prod_{n=1}^{N} \rho(\mathbf{Y_n} | \boldsymbol{\theta}, \widehat{T}(\cdot, t_{n-1}), LBC_n) =
\prod_{n=1}^{N}  \rho(\mathbf{Y_n^{I}} | \boldsymbol{\theta}, \widehat{T}(\cdot, t_{n-1}), LBC_n) \times \rho(\mathbf{Y_n^{B}} | LBC_n)\]
and finally, using \eqref{eq:dens}, the expression \eqref{eq:globlike} is obtained.
\qed
\end{proof}

\begin{remark}
A generalization of the joint likelihood can be obtained given serial correlations, that is, $\left\{\boldsymbol{\epsilon_n}\right\}_{n=1}^{N}$ are time correlated.
\end{remark}

Using the notation $\mathbf{T}_{L} = (T_{L,1}, \ldots, T_{L,N})^{tr}$, $\mathbf{T}_{R} = (T_{R,1}, \ldots, T_{R,N})^{tr}$, $\mathbf{Y}_{L} = (Y_{L,1}, \ldots, Y_{L,N})^{tr}$ and $\mathbf{Y}_{R} = (Y_{R,1}, \ldots, Y_{R,N})^{tr}$, the joint likelihood function \eqref{eq:globlike} can be written as
\begin{eqnarray}
&\!& \rho(\mathbf{Y_1}, \ldots, \mathbf{Y_N} | \boldsymbol{\theta}, \mathbf{T}_{L},\mathbf{T}_{R}) = (\sqrt{2 \pi} \sigma)^{- N (I+1)} \exp \left( - \frac{1}{2 \sigma^2}  \sum_{n=1}^{N} \left\| {\mathbf{R}}_{t_{n}}\right\|^{2}_{\ell^2} \right) \nonumber \\
&\!& \times \exp \left( - \frac{1}{2 \sigma^2} \left[ \left\| \mathbf{T}_{L} - \mathbf{Y}_{L} \right\|^{2}_{\ell^2} + \left\| \mathbf{T}_{R} - \mathbf{Y}_{R} \right\|^{2}_{\ell^2} \right] \right),
\label{eq:lik}
\end{eqnarray}
which is a suitable expression to derive later the marginal likelihood of $\boldsymbol{\theta}$.
The prior distributions for $\boldsymbol{\theta}$ and the boundary parameters $\mathbf{T}_{L}, \mathbf{T}_{R}$ can be specified in different ways. Generally speaking, the prior distribution for $\boldsymbol{\theta}$ should be proposed according to the physical properties described by the unknown parameters and taking into account available prior knowledge about the observed phenomena. As an example, it is known that the thermal conductivity of polymethyl methacrylate (plexiglas) is in the range 0.167-0.25 W/(mK). A uniform prior distribution for the thermal conductivity could be chosen if there is no preference among the values of the interval.\\
As per the boundary parameters, we may assume independent Gaussian prior distributions $T_{L,n} \sim {\mathcal{N}}(\mu_{L,n}, \sigma_p^2)$, $T_{R,n} \sim {\mathcal{N}}(\mu_{R,n}, \sigma_p^2), \:\: n=1, \ldots, N\,,$ that are centered at a least square spline fit, say $\mathop{\boldsymbol{\mu}_L}\limits_{N \times 1} = (\mu_{L,1}, \ldots, \mu_{L,N})^{tr}$,
$\mathop{\boldsymbol{\mu}_R}\limits_{N \times 1} = (\mu_{R,1}, \ldots, \mu_{R,N})^{tr}$, of the observed data, $\mathbf{Y}_L$ and $\mathbf{Y}_R$, respectively, and for some prior variance $\sigma_p^2 >0$.


We claim that the data residual vector $\mathbf{R}_{t_n}$ can be written as a linear function of $\mathbf{T}_L$ and $\mathbf{T}_R$. The next section is devoted to the proof of this basic result. This proof allows for the exact marginalization of the contribution of the nuisance boundary parameters from the joint likelihood function (\ref{eq:lik}).


\section{Numerical Approximation}
\label{sec3}

In this section, our goal is to approximate the residual vector $\mathbf{R}_{t_{n}}$ as a linear function of the boundary conditions. After reformulating the main problem \eqref{eq:main1}, as described in the next lemma, we will introduce its weak form and finite element method will be then applied to provide a numerical approximation of the solution of the weak problem.
\begin{lemma}
The solution of \eqref{eq:main1} can be written in the form
\begin{equation}
T(x,t) = T_L(t) \frac{x_R-x}{x_R - x_L} + T_R(t) \frac{x-x_L}{x_R - x_L} + u(x,t),
\label{eq:lemma}
\end{equation}
where $u$ solves a new initial-boundary value problem with homogeneous Dirichlet boundary conditions:
\begin{equation}
\begin{cases}
\partial_t u + L_{\boldsymbol{\theta}}u = f(x,t), & x \in (x_L, x_R), 0 < t \leqslant t_N < \infty \\
u(x_L,t) = 0, & t \in [0,t_N] \\
u(x_R,t) = 0, & t \in [0,t_N] \\
u(x,0) = g^{0}(x), & x \in(x_L, x_R)
\end{cases}
\label{eq:hom}
\end{equation}
and \[ f(x,t) =  - \left(\partial_t + L_{\boldsymbol{\theta}}\right) T_L(t) \frac{x_R - x}{x_R - x_L} - \left( \partial_t + L_{\boldsymbol{\theta}}\right) T_R(t) \frac{x-x_L}{x_R - x_L} \,,\]
\[ g^{0}(x) = g(x)-T_L(0) \frac{x_R-x}{x_R - x_L} - T_R(0) \frac{x-x_L}{x_R - x_L} \,. \]
\end{lemma}


We now introduce the weak formulation of problem \eqref{eq:hom}.\\
Find $u(t) \in V = H_{0}^{1}(x_L,x_R)$, $t \in (0,t_N)$ such that:
\begin{equation}
\begin{cases}
\int_{x_L}^{x_R} \partial_t u(t) v dx +  B\left( u(t) ,v \right)  = \int_{x_L}^{x_R} f(t) v dx , \forall v \in V, \,t \in (0,t_N), \\
u(0) = g^0, \: x\in(x_L, x_R),
\end{cases}
\label{eq:weak}
\end{equation}
where $B(u,v) = \int_{x_L}^{x_R} [ a(x) \partial_x u \: \partial_x v + b(x) \partial_x u \:v + c(x) u v ]\,dx$ and $ H_{0}^{1}(x_L,x_R)$ is the closure of the space $C^1_c(x_L,x_R)$ of continuously differentiable functions with compact support on $(x_L,x_R)$ with respect to the $H^1$-norm (\cite{claes}, [p.149]).

Given a mesh $ x_L = x_0 < ... < k \Delta x = x_k < ... < I\Delta x = x_{I} = x_R$ of the spatial domain $(x_L, x_R)$, we apply the finite element method with piecewise linear functions (hat functions) $\{ \phi_k \}_{k=1}^{I-1}$ to approximate the weak solution $u(x,t)$ of \eqref{eq:hom} as linear combinations of the basis functions:
\begin{equation*}
u(x,t) \approx u_{\Delta x}(x,t) = \sum_{k=1}^{I-1} u_{k}(t) \phi_{k}(x)\,,\:\:0 < t \leqslant t_N\,,\\
\end{equation*}
and we get
\begin{equation}
\sum_{k=1}^{I-1} \partial_t  u_{k}(t) \int_{x_L}^{x_R} \phi_{k} \phi_{j} dx + \sum_{k=1}^{I-1} u_{k}(t) B\left( \phi_{k}, \phi_{j} \right)  = \int_{x_L}^{x_R} f(t) \phi_{j} dx , j = 1,...I-1, \,t \in (0,t_N)\,. \label{eq:weaktwo}
\end{equation}

Let $\mathbf{u}(t) = (u_{1}(t), \ldots, u_{I-1}(t))^{tr}$, then we can write the linear system of ODEs \eqref{eq:weaktwo} in a matrix form as:
\begin{eqnarray}
\label{eq:mat1}
M \partial_t \mathbf{u}(t) +  S_{\boldsymbol{\theta}} \mathbf{u}(t)  =\mathbf{f}(t), \,  \:\:0 < t \leqslant t_N \,,
\end{eqnarray}
where $M$ is the mass matrix, $S_{\boldsymbol{\theta}}$ is the stiffness matrix and  $\mathop{\mathbf{f}(t)}$ is the load vector.
 
We consider now a uniform time discretization of $(0,t_N)$ such that
$t_0 = 0 < ... < n \Delta t = t_n < ... < N \Delta t = t_N$ and denote $\mathop{\mathbf{u}(t_n)} = \mathop{\mathbf{u}_{n}}$ and $\mathop{\mathbf{f}(t_n)} = \mathop{\mathbf{f}_{n}}$. 
We apply the backward Euler method on equation \eqref{eq:mat1} to obtain the fully discrete analogue of \eqref{eq:weak} that takes the form
\begin{equation}
\begin{cases}
\left( M + \Delta t S_{\boldsymbol{\theta}} \right) \mathbf{u}_{n+1} = M \mathbf{u}_{n} + \Delta t \mathbf{f}_{n+1} \,, n = 0, \ldots, N-1, \\
M \mathbf{u}_{0} = \mathbf{g}^0,
\end{cases}
\end{equation}
where $\mathbf{g}^{0} = \left( \int_{x_L}^{x_R} g^0 \phi_{1} dx, \ldots, \int_{x_L}^{x_R} g^0 \phi_{I-1} dx \right)^{tr}$.

\begin{theorem}
\label{th:linear}
The approximation of the weak solution of \eqref{eq:hom} can be written as a linear function of the initial and boundary conditions:
\begin{equation}
\mathbf{u}_{n} = A_{n}(\boldsymbol{\theta}) \mathbf{u}_0 + \tilde{A}_{L,n}(\boldsymbol{\theta}) {\mathbf{T}}_L + \tilde{A}_{R,n}(\boldsymbol{\theta}) {\mathbf{T}}_{R},
\label{eq:linear}
\end{equation}
where
$\mathop{\mathbf{u}_{n}} = (u_{1,n}, \ldots, u_{I-1,n})^{tr}$,
$\mathop{\mathbf{T}_{L}} = (T_{L,1}, \ldots, T_{L,N})^{tr}$,
$\mathop{\mathbf{T}_{R}} = (T_{R,1}, \ldots, T_{R,N})^{tr}$, and 
the matrices $A_{n}(\boldsymbol{\theta}), \tilde{A}_{L,n}(\boldsymbol{\theta})$ and $\tilde{A}_{R,n}(\boldsymbol{\theta})$ 
are explicitly constructed in the proof.
\end{theorem}

\begin{proof}
See Appendix A (\ref{appA}).
\end{proof}

\begin{theorem}
\label{th:linear2}
The approximation of the weak solution of \eqref{eq:main1} can be written as a linear function of the initial and boundary conditions:
\begin{equation}
\mathbf{T}_{n} = \mathbf{B}^{n} \mathbf{T}_0 + A_{L,n}(\boldsymbol{\theta}) {\mathbf{T}}_L + A_{R,n}(\boldsymbol{\theta}) {\mathbf{T}}_{R}.
\label{eq:mainsol}
\end{equation}
where $\mathbf{T}_{n}$ is defined similarly to $\mathbf{u}_{n}$ and the matrices $\mathbf{B}, A_{L,n}(\boldsymbol{\theta})$ and $A_{R,n}(\boldsymbol{\theta})$ are explicitly constructed in the proof.
\end{theorem}

\begin{proof}
See Appendix B (\ref{appB}).
\end{proof}

\begin{corollary}
The data residual vector $\mathbf{R}_{t_n} = (\widehat{T}(x_1,t_n) -  Y_{1,n},\, \ldots,\, \widehat{T}(x_{I-1},t_n) - Y_{I-1,n})^{tr}$ is approximated by
\begin{equation}
\tilde{\mathbf{R}}_{t_{n}} = \left( \mathbf{B}^{n}\mathbf{T}_{0} - \mathbf{Y_n^{I}} \right) + A_{L,n}(\boldsymbol{\theta})\mathbf{T}_L +
A_{R,n}(\boldsymbol{\theta})\mathbf{T}_R.
\label{eq:res}
\end{equation}
\end{corollary}

\section{The marginal likelihood of $\boldsymbol\theta$}
\label{sec4}

Given the observations $\mathbf{Y_1}, \ldots, \mathbf{Y_N}$, we showed, in the Section \ref{sec2}, how to obtain the joint likelihood function of $\boldsymbol{\theta}$ and the boundary parameters $\mathbf{T}_{L}, \mathbf{T}_{R}$.

In the present section, we derive a convenient expression for the marginal likelihood of $\boldsymbol\theta$, under the assumption that the prior distributions for $\mathbf{T}_{L}$ and  $\mathbf{T}_{R}$ are independent Gaussian:
\begin{equation}
\begin{aligned}
\label{bprior}
\mathbf{T}_{L} \sim \mathcal{N}(\boldsymbol{\mu}_L,\sigma_p^2 \mathbf{I}_{N}),\\
\mathbf{T}_{R} \sim \mathcal{N}(\boldsymbol{\mu}_R,\sigma_p^2 \mathbf{I}_{N}),
\end{aligned}
\end{equation}
where $\boldsymbol{\mu}_L$ and $\boldsymbol{\mu}_R$ are least square spline fits of the observed data, $\mathbf{Y}_L$ and $\mathbf{Y}_R$, respectively. Then, using \eqref{eq:lik}, the marginal likelihood of $\boldsymbol\theta$ is given by:
\begin{eqnarray}
&& \rho(\mathbf{Y_1}, \ldots, \mathbf{Y_N} | \boldsymbol{\theta}) = (\sqrt{2 \pi} \sigma)^{- N (I+1)} \times (\sqrt{2\pi}\sigma_p)^{-2N} \int_{\mathcal{T}_R} \int_{\mathcal{T}_L} \exp \left( - \frac{1}{2 \sigma^2}  \sum_{n=1}^{N} \left\| {\mathbf{R}}_{t_{n}}\right\|^{2}_{\ell^2} \right) \nonumber \\
&& \times \!\exp \!\left(\!- \frac{1}{2 \sigma^2} (\mathbf{T}_L - \mathbf{Y}_L)^{tr} (\mathbf{T}_L -\mathbf{Y}_L) - \frac{1}{2 \sigma^2} (\mathbf{T}_R -\mathbf{Y}_R)^{tr} (\mathbf{T}_R - \mathbf{Y}_R) \! \right) \nonumber \\
&& \times \!\exp \!\left(\!- \frac{1}{2 \sigma_p^2} (\mathbf{T}_L - \boldsymbol{\mu}_L)^{tr} (\mathbf{T}_L - \boldsymbol{\mu}_L) - \frac{1}{2 \sigma_p^2} (\mathbf{T}_R - \boldsymbol{\mu}_R)^{tr} (\mathbf{T}_R - \boldsymbol{\mu}_R)\! \right) d\mathbf{T}_L d\mathbf{T}_R. \label{eq:intlike}
\end{eqnarray}

To provide the exact expression of the marginal likelihood of the linear parabolic equation coefficients $\boldsymbol{\theta}$, which has been implemented in the computational examples presented in Section \ref{sec5}, it is convenient to introduce the following notation:

\[ \mathop{\Delta_L}\limits_{N \times N} = \sum_{n=1}^{N} A_{L,n}(\boldsymbol{\theta})^{tr} A_{L,n}(\boldsymbol{\theta}) \,, \:\: \mathop{\Delta_R}\limits_{N \times N}
= \sum_{n=1}^{N} A_{R,n}(\boldsymbol{\theta})^{tr} A_{R,n}(\boldsymbol{\theta})\,, \]
\[  \mathop{\Delta_{2,L}}\limits_{N \times 1} = \sum_{n=1}^{N} A_{L,n}(\boldsymbol{\theta})^{tr} (\mathbf{Y_n^{I}} - \mathbf{B}^{n}\mathbf{T}_{0})\,, \:\:
 \mathop{\Delta_{2,R}}\limits_{N \times 1} = \sum_{n=1}^{N} A_{R,n}(\boldsymbol{\theta})^{tr} (\mathbf{Y_n^{I}} - \mathbf{B}^{n}\mathbf{T}_{0})\,, \]

\[ \mathop{A_{LR}}\limits_{N \times N} = \sum_{n=1}^{N} A_{L,n}(\boldsymbol{\theta})^{tr} A_{R,n}(\boldsymbol{\theta})\,,
\:\: \mathop{D_{\sigma^2}}\limits_{N \times N} = \textrm{diag}\left(\frac{1}{\sigma^2}, \ldots, \frac{1}{\sigma^2}\right)\,, \:\: \mathop{D_{\sigma_p^2}}\limits_{N \times N} = \textrm{diag}\left(\frac{1}{\sigma_p^2}, \ldots, \frac{1}{\sigma_p^2}\right) \,.
\]
 
\begin{theorem}
\label{Th:marg}
The marginal likelihood of $\boldsymbol{\theta}$ is given by:
\[ \rho(\mathbf{Y_1}, \ldots, \mathbf{Y_N} | \boldsymbol{\theta}) = (\sqrt{2 \pi} \sigma)^{- N (I+1)}(\sqrt{2\pi}\sigma_p)^{-2N}(2 \pi)^{N/2} |\Lambda_0|^{1/2} (2 \pi)^{N/2} |\Lambda_1|^{1/2} \]
\[ \times \exp \Bigg\{ - \frac{1}{2 \sigma_p^2} \left[ {\boldsymbol{\mu}_L}^{tr} {\boldsymbol{\mu}_L} + {\boldsymbol{\mu}_R}^{tr} {\boldsymbol{\mu}_R} \right] - \frac{1}{2 \sigma^2} \left[ \mathbf{Y}_L^{tr} \mathbf{Y}_L + \mathbf{Y}_R^{tr} \mathbf{Y}_R + \sum_{i=1}^{N} (\mathbf{Y_n^{I}} - \mathbf{B}^{n}\mathbf{T}_{0})^{tr} (\mathbf{Y_n^{I}} - \mathbf{B}^{n}\mathbf{T}_{0}) \right] \]
\[ +  \frac{1}{2} ({\boldsymbol{\mu}_L}^{tr} D_{\sigma_p^2} + \mathbf{Y}_L^{tr} D_{\sigma^2} + \Delta_{2,L}^{tr} D_{\sigma^2})
\Lambda_0 (D_{\sigma_p^2} {\boldsymbol{\mu}_L} + D_{\sigma^2} \mathbf{Y}_L + D_{\sigma^2} \Delta_{2,L}) \]
\begin{equation}
 + \frac{1}{2} \left[ \mathbf{t}^{tr}_{R,2} \Lambda_1 \mathbf{t}_{R,2} + 2 \mathbf{t}^{tr}_{R,3} \Lambda_1 \mathbf{t}_{R,2} + \mathbf{t}^{tr}_{R,3} \Lambda_1 \mathbf{t}_{R,3} \right] \Bigg\},
\label{eq:marg}
\end{equation}
where $\Lambda_0, \Lambda_1, \mathbf{t}_{R,2}$ and $\mathbf{t}_{R,3}$ are independent of $\mathbf{T}_L$ and $\mathbf{T}_R$.
\end{theorem}

\begin{proof}
First, we observe that \eqref{eq:intlike} can be written as:
\[ (\sqrt{2 \pi} \sigma)^{- N (I+1)} \times (\sqrt{2\pi}\sigma_p)^{-2N} \times \exp \Bigg\{ - \frac{1}{2 \sigma^2_p} \big[ {\boldsymbol{\mu}_L}^{tr} {\boldsymbol{\mu}_L} + {\boldsymbol{\mu}_R}^{tr} {\boldsymbol{\mu}_R} \big] \]
\[ - \frac{1}{2 \sigma^2} \left[  \mathbf{Y}_L^{tr} \mathbf{Y}_L + \mathbf{Y}_R^{tr} \mathbf{Y}_R + \sum_{n=1}^{N} (\mathbf{Y_n^{I}} - \mathbf{B}^{n}\mathbf{T}_{0})^{tr} (\mathbf{Y_n^{I}} - \mathbf{B}^{n}\mathbf{T}_{0}) \right] \Bigg\} \]
\[ \times \int_{\mathcal{T}_R} \int_{\mathcal{T}_L} \exp \Bigg\{ - \frac{1}{2}  \left[ \mathbf{T}_L^{tr} (D_{\sigma^2} + D_{\sigma_p^2} + \frac{1}{\sigma^2} \Delta_L) \mathbf{T}_L - 2({\boldsymbol{\mu}_L}^{tr} D_{\sigma_p^2} + \mathbf{Y}_L^{tr} D_{\sigma^2} + \Delta_{2,L}^{tr} D_{\sigma^2})\mathbf{T}_L \right. \]
\[ \left. + \frac{2}{\sigma^2} \mathbf{T}_R^{tr} A_{LR}^{tr} \mathbf{T}_L + \mathbf{T}_R^{tr} (D_{\sigma^2} + D_{\sigma_p^2} + \frac{1}{\sigma^2} \Delta_R) \mathbf{T}_R - 2 ({\boldsymbol{\mu}_R}^{tr} D_{\sigma_p^2} + \mathbf{Y}_R^{tr} D_{\sigma^2} + \Delta_{2,R}^{tr} D_{\sigma^2}) \mathbf{T}_R
 \right]  \Bigg\}d\mathbf{T}_L d\mathbf{T}_R \,. \] 
 
To marginalize $\mathbf{T}_L$ and $\mathbf{T}_R$, we assume that they are independent Gaussian random vectors. We can then use the following standard result:
if $\mathbf{X} \sim {\mathcal{N}}_{p} (\boldsymbol{\mu}, \mathbf{\Sigma})$, then $ E (\exp(\mathbf{t}^{tr} \mathbf{X}))
= \exp(\mathbf{t}^{tr} \boldsymbol{\mu} + \frac{1}{2} \mathbf{t}^{tr} \mathbf{\Sigma} \mathbf{t} ) \,.$

Therefore, by integrating first with respect to $\mathbf{T}_L$, the marginal likelihood of $\boldsymbol{\theta}$ and $\mathbf{T}_R$ is proportional to the product of a factor that is independent of $\mathbf{T}_L$ and the following term

\[ \int_{\mathcal{T}_L} \exp \left\{ -\frac{1}{2} \mathbf{T}_L^{tr} \left( D_{\sigma^2} + D_{\sigma_p^2} + \frac{1}{\sigma^2} \Delta_L \right) \mathbf{T}_L \right\} \, \exp(\mathbf{t}^{tr}_{L,1} \mathbf{T}_L) d \mathbf{T}_L\,,   \]

where \[ \mathop{\mathbf{t}^{tr}_{L,1}}\limits_{1 \times N} = ({\boldsymbol{\mu}_L}^{tr} D_{\sigma_p^2} + \mathbf{Y}_L^{tr} D_{\sigma^2} + \Delta_{2,L}^{tr} D_{\sigma^2}) -  \frac{1}{\sigma^2} \mathbf{T}_R^{tr} A_{LR}^{tr} \,.\]

It is now convenient to define $\Lambda_0^{-1} := (D_{\sigma^2} + D_{\sigma_p^2} + \frac{1}{\sigma^2} \Delta_L)$. The marginal likelihood of $\boldsymbol{\theta}$ and $\mathbf{T}_R$ is proportional to the product of a factor that is independent of $\mathbf{T}_L$ and the term
$(2 \pi)^{N/2} |\Lambda_0|^{1/2} \exp \left\{ \frac{1}{2} \mathbf{t}^{tr}_{L,1} \Lambda_0 \mathbf{t}_{L,1} \right\}\,.$ \\
Therefore, the marginal likelihood of $\boldsymbol{\theta}$ can be explicitly written as
\begin{eqnarray*}
&& (\sqrt{2 \pi} \sigma)^{- N (I+1)} \times (\sqrt{2\pi}\sigma_p)^{-2N} \times (2 \pi)^{N/2} |\Lambda_0|^{1/2} \\
& \times & \exp \Bigg\{ - \frac{1}{2 \sigma_p^2} \big( {\boldsymbol{\mu}_L}^{tr} {\boldsymbol{\mu}_L} + {\boldsymbol{\mu}_R}^{tr} {\boldsymbol{\mu}_R} \big) - \frac{1}{2 \sigma^2} \bigg[ \mathbf{Y}_L^{tr} \mathbf{Y}_L + \mathbf{Y}_R^{tr} \mathbf{Y}_R + \sum_{i=1}^{N} (\mathbf{Y_n^{I}} - \mathbf{B}^{n}\mathbf{T}_{0} )^{tr} (\mathbf{Y_n^{I}} - \mathbf{B}^{n}\mathbf{T}_{0}) \bigg] \Bigg\} \\
&  \times & \int_{\mathcal{T}_R} \exp \Bigg\{ \frac{1}{2} \mathbf{t}^{tr}_{L,1} \Lambda_0 \mathbf{t}_{L,1} - \frac{1}{2}  \Bigg[ \mathbf{T}_R^{tr} (D_{\sigma^2} + D_{\sigma_p^2} + \frac{1}{\sigma^2} \Delta_R) \mathbf{T}_R -
2 ({\boldsymbol{\mu}_R}^{tr} D_{\sigma_p^2} + \mathbf{Y}_R^{tr} D_{\sigma^2} + \Delta_{2,R}^{tr} D_{\sigma^2})\mathbf{T}_R  \Bigg]  \Bigg\} d\mathbf{T}_R \,.
\end{eqnarray*}

The entire last expression is equal to the product of a term that is independent of $\mathbf{T}_R$ and the following term:
\[ \int_{\mathcal{T}_R} \exp \left\{ - \frac{1}{2} \mathbf{T}_R^{tr} \left[ (D_{\sigma^2} + D_{\sigma_p^2} + \frac{1}{\sigma^2} \Delta_R) - \left(\frac{1}{\sigma^2}\right)^2 A_{LR}^{tr} \Lambda_0 A_{LR} \right] \mathbf{T}_R \right\} \, \exp\left\{\mathbf{t}^{tr}_{R,1} \mathbf{T}_R\right\} d\mathbf{T}_R \,,  \]
where

\[ \mathbf{t}^{tr}_{R,1} = \underbrace{ ({\boldsymbol{\mu}_R}^{tr} D_{\sigma_p^2} + \mathbf{Y}_R^{tr} D_{\sigma^2} + \Delta_{2,R}^{tr} D_{\sigma^2})}_{\mathbf{t}^{tr}_{R,2}}
 \underbrace{- ({\boldsymbol{\mu}_L}^{tr} D_{\sigma_p^2} + \mathbf{Y}_L^{tr} D_{\sigma^2} + \Delta_{2,L}^{tr} D_{\sigma^2}) \Lambda_0 A_{LR}  }_{\mathbf{t}^{tr}_{R,3}} \,.   \]

If we now define $\Lambda_1^{-1} := (D_{\sigma^2} + D_{\sigma_p^2} + \frac{1}{\sigma^2} \Delta_R) - \left(\frac{1}{\sigma^2}\right)^2 A_{LR}^{tr} \Lambda_0 A_{LR}$ and integrate with respect to $ \mathbf{T}_R$, we have 
\begin{eqnarray*}
&& \int_{\mathcal{T}_R} \exp \left\{ -\frac{1}{2} \mathbf{T}_R^{tr} \Lambda_1^{-1} \mathbf{T}_R \right\}  \exp(\mathbf{t}^{tr}_{R,1} \mathbf{T}_R) d \mathbf{T}_R = (2 \pi)^{N/2} |\Lambda_1|^{1/2} \exp\left\{ \frac{1}{2} \mathbf{t}^{tr}_{R,1} \Lambda_1 \mathbf{t}_{R,1}  \right\}\\
&& = (2 \pi)^{N/2} |\Lambda_1|^{1/2} \exp\left\{ \frac{1}{2} (\mathbf{t}^{tr}_{R,2} + \mathbf{t}^{tr}_{R,3}) \Lambda_1 (\mathbf{t}_{R,2} + \mathbf{t}_{R,3}) \right\},
\end{eqnarray*}
after evaluating the integral. We finally obtain \eqref{eq:marg}.
\qed
\end{proof}

\section{A Bayesian inference for thermal diffusivity}
\label{sec5}

In this section, we implement our Bayesian approach to infer the thermal diffusivity $\theta$, an unknown parameter that appears in the heat equation and measures the rapidity of the heat propagation through a material \citep{Wilson}. Temperature data are available on the basis of cooling experiments. Synthetic data are used to carry out the inference. 

Consider the heat equation (one--dimensional diffusion equation for $T(x,t)$):
\begin{equation}
\begin{cases}
\partial_t T - \partial_x \left(\theta(x) \partial_x T \right) = 0, & x \in (x_L, x_R), \, 0 < t \leqslant t_N < \infty \\
T(0,t) = T_L(t), & t \in [0,t_N] \\
T(1,t) = T_R(t), & t \in [0,t_N] \\
T(x,0) = g(x), & x \in(x_L, x_R).
\end{cases}
\label{eq:main}
\end{equation}

We want to infer the thermal diffusivity, $\theta(x)$, using a Bayesian approach when the temperature is measured at $I+1$ locations, $x_0=x_L, x_1, x_2, \ldots, x_{I-1}, x_I = x_R$,
at each of the $N$ times, $t_1, t_2, \ldots, t_N$. Clearly, this problem is a special case of \eqref{eq:main1} where $L_{\boldsymbol{\theta}} = - \partial_x \left( \theta(x) \partial_x T \right)$ and $\theta(x) > 0$. We can therefore immediately obtain the non-normalized posterior distribution of $\theta$ using the marginal likelihood \eqref{eq:marg}.

The prior distributions for $\theta$ can be specified in different ways. In this section we will consider two cases, when $\theta$ is a lognormal random variable and when $\theta$ depends on the space variable $x$ and is modeled by means of a lognormal random field. We focus on the second case in Subsection \ref{sectionfivethree}. We start by discussing the case where the thermal diffusivity prior is independent of $x$.

If we consider a lognormal prior $\log{\theta} \sim \mathcal{N} \left( \nu, \tau \right)$, where $\nu \in \mathbb{R}$ and $\tau > 0$, then the non-normalized posterior distribution of $\theta$ is given by
\begin{equation}
\rho_{\nu, \tau}(\theta | \mathbf{Y_1}, \ldots, \mathbf{Y_N} ) \propto \frac{1}{\sqrt{2 \pi} \theta \tau} \exp \left( - \frac{(\log  \theta- \nu)^2}{2 \tau^2} \right) \rho(\mathbf{Y_1}, \ldots, \mathbf{Y_N} | \theta).
\label{eq:postfin}
\end{equation}
The posterior distribution of $\theta$ can be approximated by Laplace's method (\cite{Ghosh}, [Chapter 4]) to obtain a Gaussian posterior
\begin{equation*}
\rho_{\nu, \tau}(\theta| \mathbf{Y_1}, \ldots, \mathbf{Y_N}) \approx \frac{1}{\sqrt{2 \pi |H(\hat{\theta})|}} \exp \left\{ -(\theta- \hat{\theta})^{tr}H(\hat{\theta})^{-1}(\theta- \hat{\theta}) \right\}
\end{equation*}
where $\hat{\theta}$ is the maximum a posteriori probability (MAP) estimate and $H(\hat{\theta})$ is the Hessian matrix of the log posterior evaluated at $\hat{\theta}$.

To assess the behavior of our method, we introduce a synthetic dataset generated with constant $\theta$. Let us assume, without loss of generality, that the interval time $[0, t_N]$ is equal to $[0, 1]$, $x_L=0$ and $x_R=1$.\\

\textbf{Dataset A}\\
In order to generate data, we solve the initial-boundary value problem for the heat equation with Robin boundary conditions:
\begin{eqnarray*}
\partial_x T(x_L,t) &=& \frac{h}{\kappa} \left( T(x_L,t) -  T_{out} \right) \,, \:\:t \in [0,1], \\
\partial_x T(x_R,t) &=& \frac{h}{\kappa} (T_{out} - T(x_R,t)) \,, \:\:t \in [0,1],
\end{eqnarray*}
and the initial condition $T(x,0) = T_0 \,,\:\:x\in(0, 1)$, where $ \theta(x) = 1 \times 10^{-7}\, m^2/s $, $h$ is the convective heat transfer coefficient, $\kappa$ denotes the thermal conductivity, $ \frac{h}{\kappa} = 1\, (1/m)$, $T_{out}= 20\,^{\circ}\mathrm{C}$ and $T_0 = 100\,^{\circ}\mathrm{C}$ (see Figure \ref{pic:exact}).

\begin{figure}[h]
\centering
\includegraphics[scale=0.55]{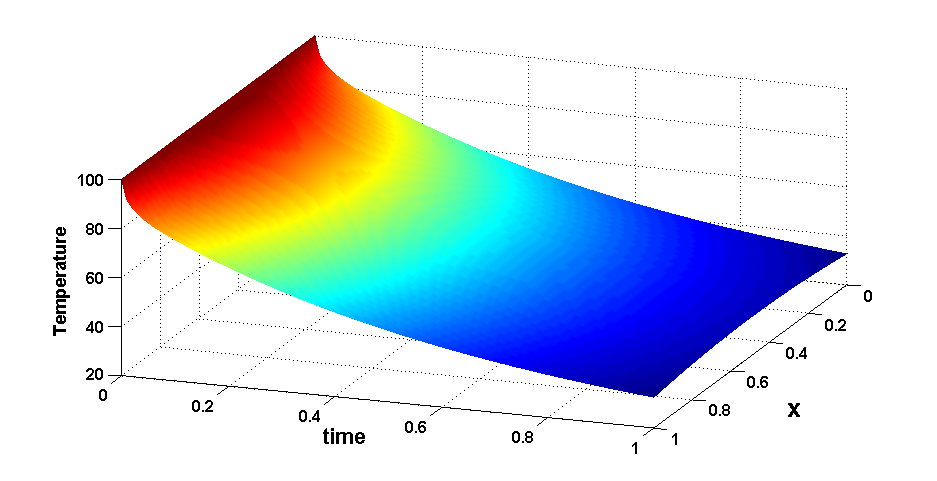}
\caption{Exact solution of the initial-boundary value problem for the heat equation, $\theta = 1 \times 10^{-7}\, m^2/s$.}
\label{pic:exact}
\end{figure}

A synthetic dataset (hereafter named dataset A) is generated, with a measurement standard error noise of $\sigma_{d} = 0.56$. \\
Before presenting the implementation of our novel technique, we will show how the Bayesian method works, using the joint likelihood \eqref{eq:lik}, under the very restrictive assumption that the temperature values at the boundaries are exactly known.

\subsection{Example 1} 

Suppose that the thermal diffusivity, $\theta$, is a random variable with a lognormal prior, $\log{\theta} \sim \mathcal{N} \left( \nu, \tau \right)$. In this case, the non-normalized posterior density for $\theta$ is given by
\begin{equation*}
\rho_{\nu, \tau}(\theta | \mathbf{Y_1}, \ldots, \mathbf{Y_N} ) \propto \frac{1}{\sqrt{2 \pi} \theta \tau} \exp \left( - \frac{(\log  \theta- \nu)^2}{2 \tau^2} \right)  \exp \left( - \frac{1}{2 \sigma^2}  \sum_{n=1}^{N} \left\| {\mathbf{R}}_{t_{n}}\right\|^{2}_{\ell^2} \right),
\label{simplepost}
\end{equation*}
where ${\mathbf{R}}_{t_{n}}$ is used here because the boundary data are known exactly.

\begin{figure}[h!]
	\begin{minipage}{.4\linewidth}
	\centering
		\includegraphics[scale=0.4]{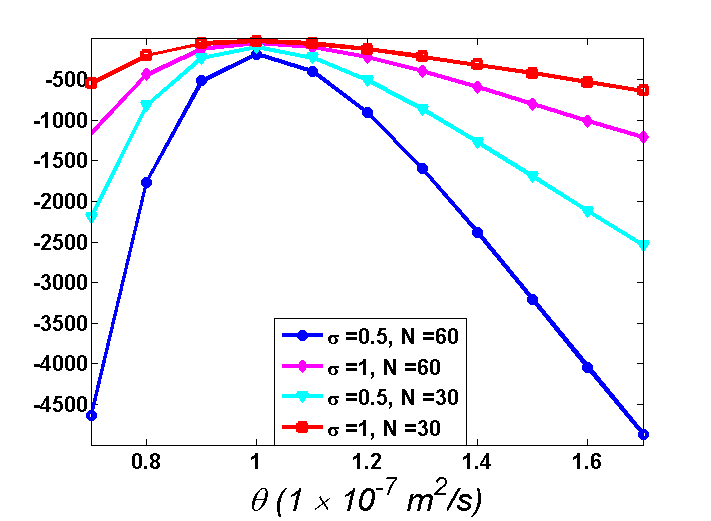}
		\end{minipage}%
	\hspace{5em}
	\begin{minipage}{.4\linewidth}
	\centering
		\includegraphics[scale=0.4]{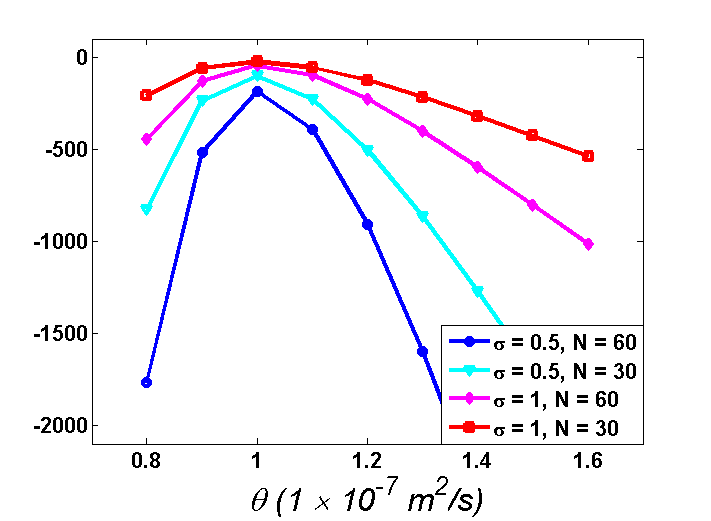}
		\end{minipage}
\caption{Example 1: Comparison between log-likelihoods (on the left) and log-posteriors (on the right) for $\theta$ using different numbers of observations, $N$, and different values of $\sigma$.}
\label{pic:Ex0}
\end{figure}

Given that $\theta$ is a lognormal random variable with $\nu = \tau = 0.1$, the resulting posterior will depend on $\sigma$ and the number of observations $N$ that are used to compute the log-likelihood. Figure \ref{pic:Ex0} shows the behavior of the log-likelihood and the log-posterior for $\theta$ using different values for $\sigma$ and $N$.

\begin{figure}[h!]
	\centering
	\includegraphics[scale=0.55]{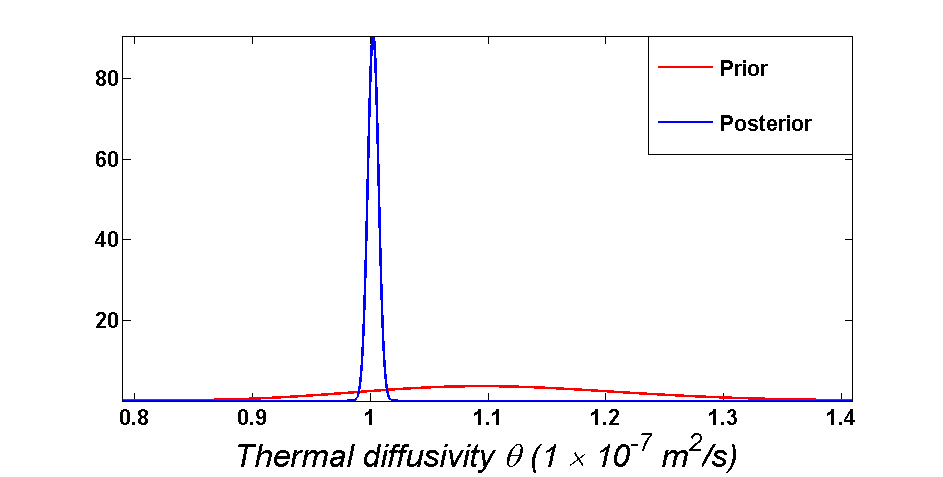}
	\caption{Example 1: Lognormal prior and approximated Gaussian posterior densities for $\theta$, where $\sigma_p = \sigma = 0.5$ and $N = 60$.}
	\label{pic:priorpostsimple}
\end{figure}

We then use the Laplace approximation to derive the Gaussian posterior approximated density for $\theta$. The prior and posterior densities for $\theta$ are presented in Figure \ref{pic:priorpostsimple}, where it can be appreciated that we obtained a Gaussian posterior, with mean $1.0025$ and standard deviation $0.0044$, which is concentrated around the true value of the parameter, $\theta$, despite having a very broad prior. We are now in the position to extend our implementation to embrace the case where the temperature values at the boundaries are unknown parameters as well.

\subsection{Example 2}

In this example, we consider again $\theta$ as a random variable with a lognormal prior, $\log{\theta} \sim \mathcal{N} \left( \nu, \tau \right)$. Unlike Example 1, we assume noisy boundary measurements and a Gaussian prior distribution for the boundary parameters as in \eqref{bprior}. Therefore, the non-normalized posterior density for $\theta$ is given by
\begin{equation*}
\rho_{\nu, \tau}(\theta | \mathbf{Y_1}, \ldots, \mathbf{Y_N} ) \propto \frac{1}{\sqrt{2 \pi} \theta \tau} \exp \left( - \frac{(\log  \theta- \nu)^2}{2 \tau^2} \right) \rho(\mathbf{Y_1}, \ldots, \mathbf{Y_N} | \theta),
\label{simplepost}
\end{equation*}
where $\rho(\mathbf{Y_1}, \ldots, \mathbf{Y_N} | \theta)$ is the marginal likelihood of $\theta$ defined in Theorem \ref{Th:marg}.

\begin{remark}
\label{re:ex2}
Since $\theta(x)$ is supposed to be constant, we can obtain equation \eqref{eq:mainsol} alternatively by solving the heat equation using finite differences (see Appendix C \ref{appC}).
\end{remark}

\begin{figure}[h!]
	\begin{minipage}{.4\linewidth}
	\centering
		\includegraphics[scale=0.4]{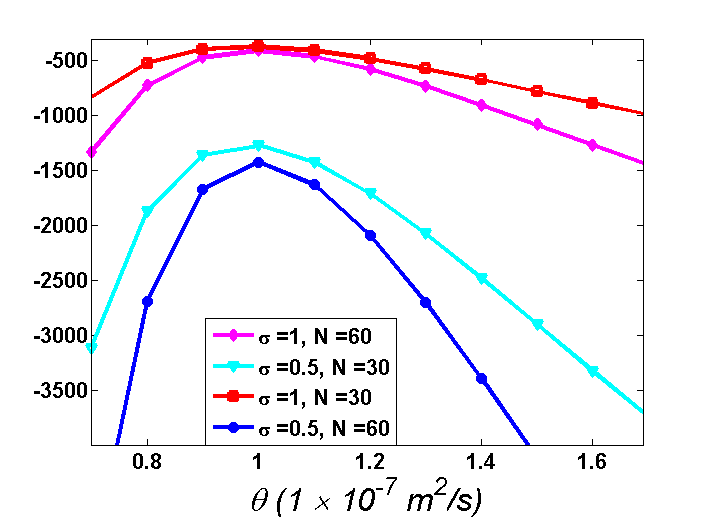}
	\end{minipage}
	\hspace{5em}
	\begin{minipage}{.4\linewidth}
	 \centering
		\includegraphics[scale=0.4]{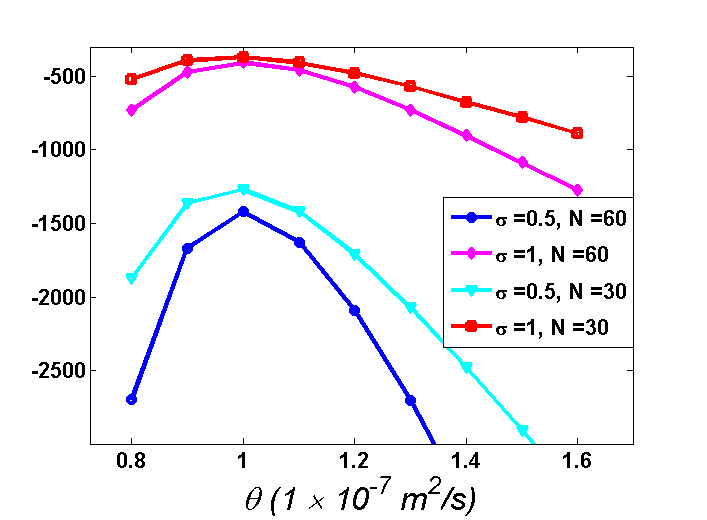}
	\end{minipage}
\caption{Example 2: Comparison between log-likelihoods (on the left) and log-posteriors (on the right) for $\theta$ using different numbers of observations, $N$, and different values of $\sigma$.}
\label{pic:Ex11}
\end{figure}

Numerical results are now presented using the synthetic dataset A and assuming that $\theta$ is a lognormal random variable with $\nu = \tau = 0.1$. Figure \ref{pic:Ex11} shows the behavior of the log-likelihood and the log-posterior for $\theta$ using different values of $N$ and $\sigma$. Clearly, the accuracy of the estimated $\theta$ depends on the size of the dataset, $N$, and the reliability of measurement devices, $\sigma$.

\begin{figure}[h!]
	\begin{minipage}{.4\linewidth}
	\centering
		\includegraphics[scale=0.4]{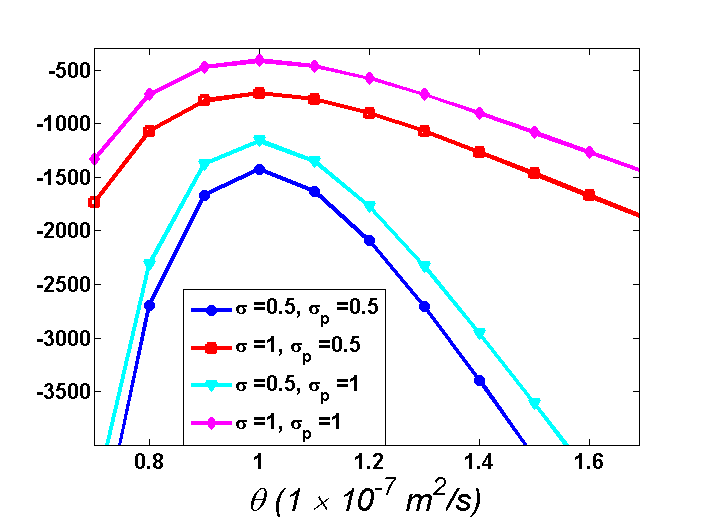}
	\end{minipage}
	\hspace{5em}
	\begin{minipage}{.4\linewidth}
	 \centering
		\includegraphics[scale=0.4]{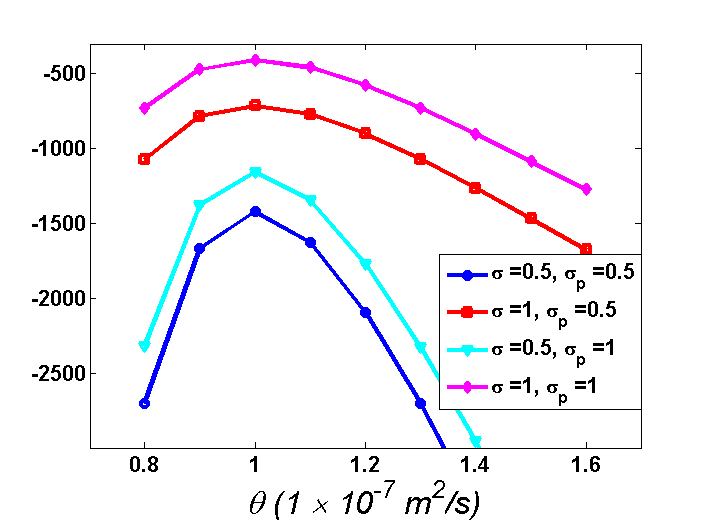}
	\end{minipage}
\caption{Example 2: Comparison between log-likelihoods (on the left) and log-posteriors (on the right) for $\theta$ using different values of $\sigma$ and $\sigma_p$, with $N = 60$.}
\label{pic:Ex12}
\end{figure}

Figure \ref{pic:Ex12} shows the relationship between the prior distribution of the boundary conditions and their measurements.  Although we notice different behaviors of the log-likelihood and the log-posterior, these functions exhibit the same argument of the maximum which is close to the true value of $\theta$.

\begin{figure}[h!]
\centering
\includegraphics[scale=0.55]{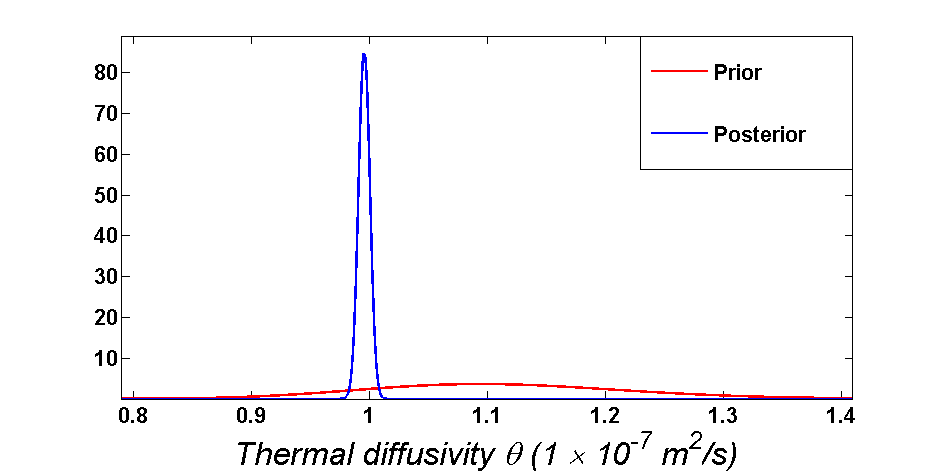}
\caption{Example 2: Lognormal prior and approximated Gaussian posterior densities for $\theta$ where $\sigma_p = \sigma = 0.5$ and $N = 60$.}
\label{pic:priorpost}
\end{figure}

Again, the Laplace approximation is used to derive the Gaussian posterior approximated density for $\theta$. The prior and posterior densities for $\theta$ are presented in Figure \ref{pic:priorpost} in which the Gaussian posterior, with mean $0.9955$ and standard deviation $0.0047$, is concentrated around the true value of $\theta$ unlike the very broad prior.

\subsubsection{Information Divergence and Expected Information Gain}

In the Bayesian setting that we adopted to infer the thermal diffusivity, $\theta$, the utility of the performed experiment, given an experimental setup $\xi$, can be conveniently measured by the so-called information divergence (or discrimination information as \cite{Kull2} called it), which is here defined as the Kullback-Leibler divergence (\cite{Kull}) between the prior density function $p(\theta)$ and the posterior density function of $\theta$, $\rho(\theta|\mathbf{Y_1}, \ldots, \mathbf{Y_N}, \xi)$:
\begin{equation}
D_{KL}(\mathbf{Y_1}, \ldots, \mathbf{Y_N}, \xi) := \int_{\Theta} \log \left( \frac{\rho(\theta| \mathbf{Y_1}, \ldots, \mathbf{Y_N}, \xi)}{p(\theta)} \right) \rho(\theta|\mathbf{Y_1}, \ldots, \mathbf{Y_N}, \xi) d\theta \,. \label{eq:kulldiv}
\end{equation}

The quantity in (\ref{eq:kulldiv}) is always non-negative; it is equal to zero when the prior and the posterior coincide; it provides a quantification of the relative discrimination between the prior and the posterior; and it depends on the observations $\mathbf{Y_1}, \ldots, \mathbf{Y_N}$. Therefore, given the synthetic dataset, A, we may introduce different experimental setups of interest by varying the interval time during which the temperature is measured. By choosing some specific thermocouples,we may evaluate the information divergence for any experimental setup.
Moreover, under the same generating process used for the dataset A, we may obtain as many synthetic datasets as needed to explore the properties of the proposed simulated experiment. The utility of such computer-based experiments can be adequately summarized by the so-called expected information gain (\cite{Quan}), which is defined as the marginalization of $D_{KL}$ over all possible simulated data:
\begin{equation} 
 I(\xi) := \int_{\cal{Y}} \int_{\Theta} \log \left( \frac{\rho(\theta| \mathbf{Y_1}, \ldots, \mathbf{Y_N}, \xi)}{p(\theta)} \right)\rho(\theta|\mathbf{Y_1}, \ldots, \mathbf{Y_N}, \xi) d\theta \rho(\{\mathbf{Y_i} \}_{i=1}^{N} |\xi)d(\{\mathbf{Y_i}\}_{i=1}^{N}). \label{eq:kulldivave}
\end{equation}
This quantity (\ref{eq:kulldivave}) provides a criterion to determine which features of the setup, $\xi$, are, on average, most informative when inferring $\theta$. A larger value of $I(\xi)$ when, say, $\xi \in A$, suggests that, given the proposed statistical model, the inference on the unknown parameter will be more efficient, on average, when the features of the designed experiment take value in the set $A$.

Let us label as ${\textrm{TC1},\ldots,\textrm{TC7}}$ the thermocouples from the left boundary to the right boundary, respectively.
 
The numerical estimations of the information divergence for the synthetic dataset A and of the expected information gain, by using  \eqref{eq:postfin} to compute the approximated posterior in Example 1, are shown in Figures \ref{pic:infotime}, \ref{pic:infothermo} and \ref{pic:infotimethermo}, which refer to the following three experimental setups (es's):
\begin{itemize}
\item[es1)] $\xi$ consists of three non-overlapping time intervals, with the same length, which cover the entire observational period $[0,1]\,;$
\item[es2)] $\xi$ consists of the five inner thermocouples;
\item[es3)] $\xi$ is the combination of the two previous experimental setups, es1 and es2.
\end{itemize}

\begin{figure}[h!]
\centering
\includegraphics[scale=0.55]{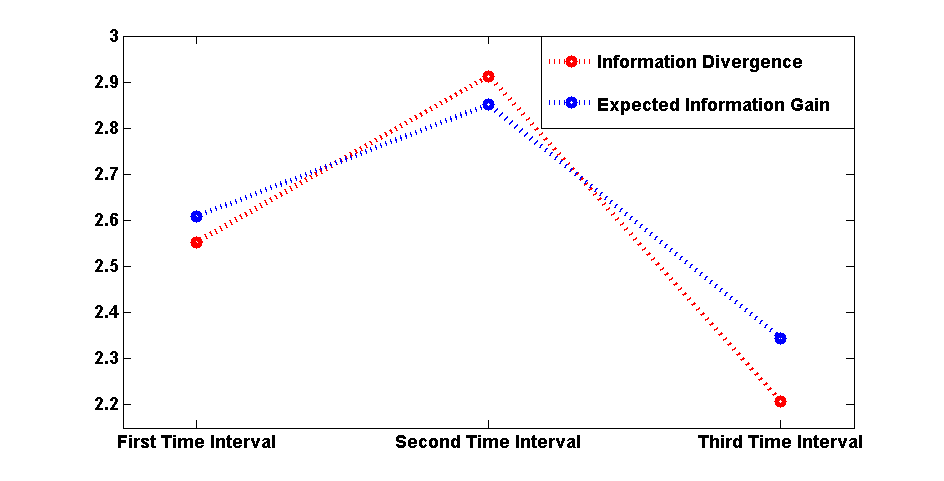}
\caption{Example 2: The expected information gain compared with the information divergence for the synthetic dataset, A, for the three time intervals experimental setup (es1).}
\label{pic:infotime}
\end{figure}

\begin{figure}[h!]
\centering
\includegraphics[scale=0.55]{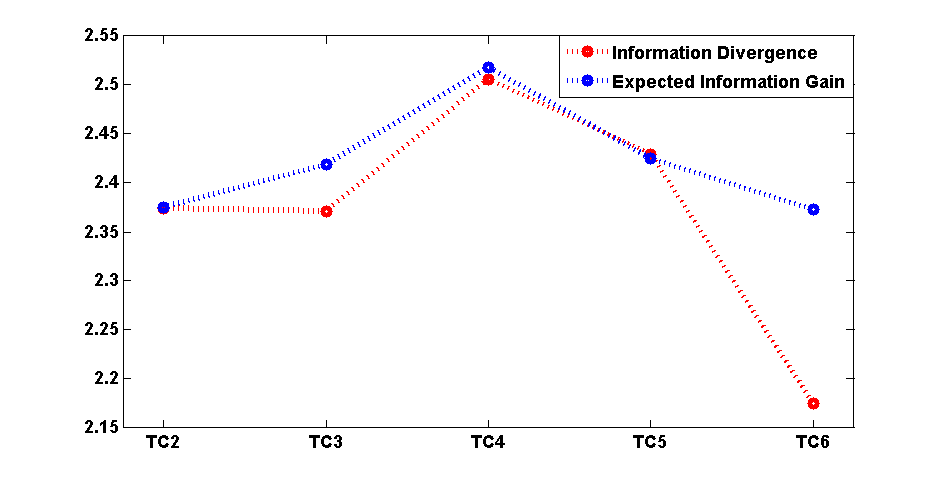}
\caption{Example 2: The expected information gain compared with the information divergence for the synthetic dataset, A, for the five inner thermocouples experimental setup (es2).}
\label{pic:infothermo}
\end{figure}

\begin{figure}[h!]
\centering
\includegraphics[scale=0.55]{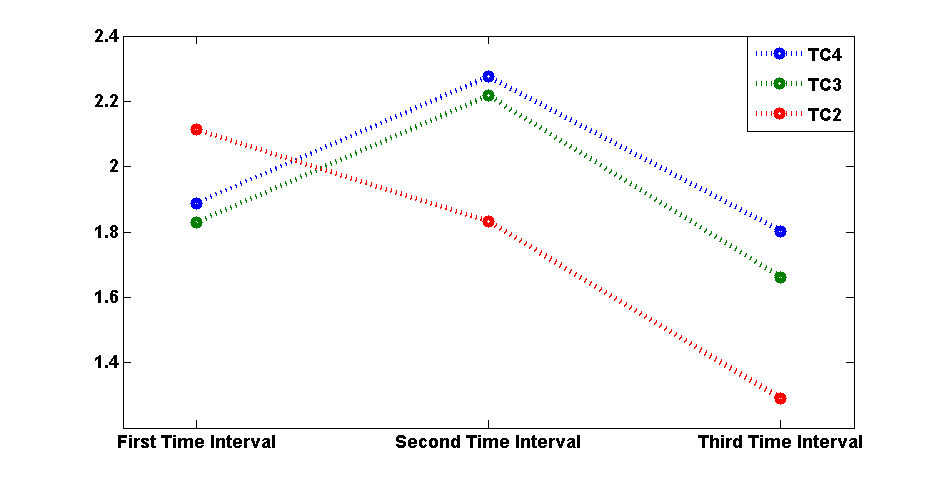}
\caption{Example 2: The expected information gain computed for the combination (es3) of the three time intervals and the five inner thermocouples experimental setups.}
\label{pic:infotimethermo}
\end{figure}

From the values in Figure \ref{pic:infotime}, which depicts the results from an experimental setup in which the temperature measurements are collected at different time intervals, we may conclude, by virtue of the interpretation of the expected information gain, that the second time interval is the most informative time interval from which to draw inferences on the thermal diffusivity, whereas the last time interval is the least informative one.

Figure \ref{pic:infothermo} summarizes how the expected information gain behaves for the five inner thermocouples experimental setup (es2). Given the synthetic dataset A, the sixth thermocouple (TC6) is the one where the information divergence takes the smallest value. However, when we look at the expected information gain, we may appreciate the nearly symmetric informative content of the thermocouples with respect to the central thermocouple (TC4) and how the expected gain about the thermal diffusivity becomes larger near the central thermocouple.

Finally, we look for the best combination of the two previous experimental setups, and the corresponding results are displayed in Figure \ref{pic:infotimethermo}.  We observe that the highest expected information gain is attained at the middle thermocouple (TC4) by using the information collected during the second time interval. Any indication provided by the numerical estimation of the expected information gain is very valuable to an experimentalist, since it suggests the most relevant features to be considered to build up an efficient experiment to infer the unknown parameters of the assumed statistical model.

\subsubsection{Predictive Posterior Distribution}

In this section, we examine the possibility of predicting the observable temperature at future time intervals after estimating the thermal diffusivity, $\theta$. More specifically, assume we have inferred $\theta$ using temperature measurements up to time $t_n$. Then, we want to predict the temperature in the next time step, $t_{n+1}$. Given our Bayesian model, it is necessary to assume the knowledge of the boundary temperature at time $t_{n+1}$. A typical situation could be given by an experiment in which there is interest in temperature values at inner points for different boundary values. 

The predictive posterior distribution, $\rho(\mathbf{Y_{n+1}} | \left\{ \mathbf{Y_k} \right\}_{k = 1}^{n}, T_{L,n+1}, T_{R,n+1} )$, is given by
\begin{equation}
\int_{\Theta} \rho(\mathbf{Y_{n+1}} | \left\{ \mathbf{Y_k} \right\}_{k = 1}^{n}, T_{L,n+1}, T_{R,n+1}, \theta ) \rho(\theta | \left\{ \mathbf{Y_k} \right\}_{k = 1}^{n}) \,d\theta \,,
\end{equation}
and it is estimated by averaging
\begin{equation*}
\frac{1}{M} \sum_{i = 1}^{M} \rho(\mathbf{Y_{n+1}} | \left\{ \mathbf{Y_k} \right\}_{k = 1}^{n}, T_{L,n+1}, T_{R,n+1}, \theta_i)\,,
\end{equation*}
where the $\theta_i$'s are sampled from the posterior distribution of $\theta$. \\

Figure \ref{fig:ppm} shows the one-step-ahead predictive posterior densities at three different inner thermocouples based on the observations until time $t=0.5$, when the observed temperature at thermocouples $\textrm{TC2},\textrm{TC3}$ and $\textrm{TC4}$ were $53.98, 55.53$ and $57.84\,^{\circ}\mathrm{C}$ respectively.
 
\begin{figure}[h!]
\centering
\includegraphics[scale=0.55]{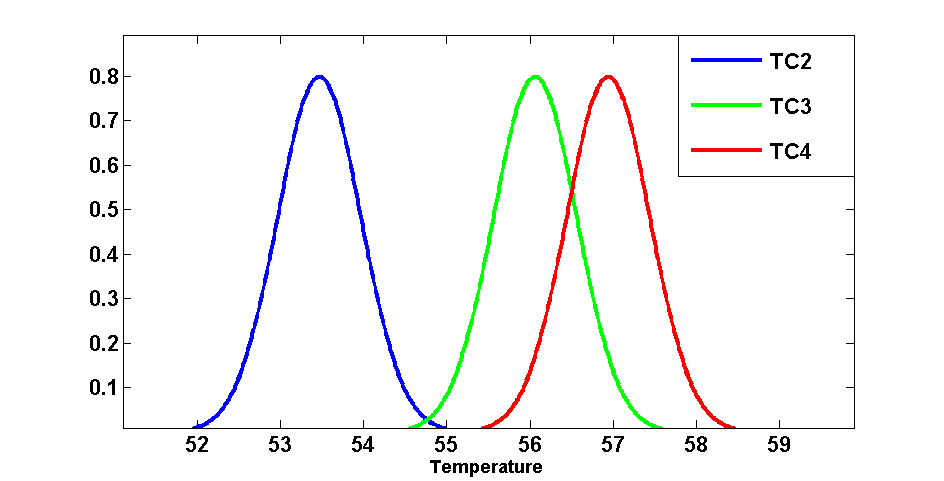}
\caption{Example 2: The predictive posterior densities of the observable temperatures for the thermocouples $\textrm{TC2}, \textrm{TC3}, \textrm{TC4}$ at time $t=0.52$.}
\label{fig:ppm}
\end{figure}

We emphasize that this methodology allows us to obtain the k-step ahead predictive posterior density for any inner thermocouple, assuming boundary conditions subject to uncertainty that are adequate for the experiment.

\subsection{Example 3} \label{sectionfivethree}

In this example, we consider the case where the thermal diffusivity depends on the space variable, $x$. The finite element method used to solve the heat equation under such an assumption was presented in Section \ref{sec3}.

\subsubsection{Prior distribution of $\theta(x)$}
Assume that the prior distribution of $\theta(x)$ is a lognormal random field with a squared exponential (SE) covariance function. Then, the prior distribution of $\log \theta(x)$ can be expressed using the joint multivariate Gaussian distribution:
\begin{equation}
\left( \log(\theta(x_1)), \ldots, \log(\theta(x_s)) \right) \sim \mathcal{N}_s \left( \boldsymbol{\mu}, K \right), \label{priormulti}
\end{equation}
where $\boldsymbol{\mu} = (\mu,\mu, ..., \mu)^{tr}$, $K_{i j} = Cov(\log(\theta(x_i)), \log(\theta(x_j))) = \eta^2 \exp \left(- \frac{\vert x_i - x_j \vert^2}{2\ell} \right)\,,i,j=1,\ldots,s$, $\eta$ is the magnitude, and $\ell$ denotes the length scale.

We assume the following priors for the hyperparameters $\mu, \eta$ and $\ell$ (the prior density of $\mu$ and $\eta$ is displayed in Figure \ref{pic:prior2D1}):
\[ \mu \sim \mathcal{N} \left( 0.1, 0.1 \right), \quad
\eta \sim \textrm{half-Cauchy}\left( 0.1 \right), \quad
\ell \sim U \left( 0.5, 5 \right). \]
In this example, we choose a Gaussian prior for $\mu$ and uninformative uniform prior for $\ell$. The half-Cauchy prior for $\eta$ was chosen because it is a practical prior for scale parameters in hierarchical models \citep{polson} \citep{gelman}.  

\begin{figure}[h!]
\centering
\includegraphics[scale=0.35]{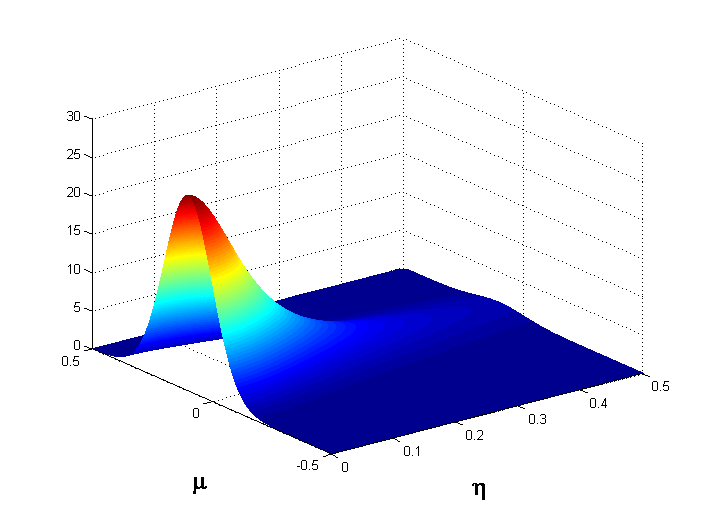}
\caption{Example 3: Joint prior density for the hyperparameters $\mu$ and $\eta$, with $\mu \sim \mathcal{N} \left( 0.1, 0.1 \right)$ and $\eta \sim \textrm{half-Cauchy}\left( 0.1 \right)$.}
\label{pic:prior2D1}
\end{figure}

\subsubsection{Joint posterior distribution of the hyperparameters $\mu,\eta, \ell$}
Given $\boldsymbol{\theta} := \left( \theta(x_1), \ldots, \theta(x_s) \right)^{tr}$, let us consider the joint posterior density of the hyperparameters $(\mu,\eta, \ell)$ that characterize the distribution of $\log \theta(x)$.
\begin{eqnarray*}
\rho(\mu , \eta, \ell |  \mathbf{Y_1}, \ldots, \mathbf{Y_N}) \propto \rho(\mu ,\eta, \ell) \int_{\boldsymbol{\Theta}} \rho(\boldsymbol{\theta}|\mu , \eta, \ell) \rho( \mathbf{Y_1}, \ldots, \mathbf{Y_N}|\boldsymbol{\theta}) d\boldsymbol{\theta}.
\end{eqnarray*}
Let us introduce the auxiliary variable $\mathbf{z} = (z_1, \dots, z_s) \sim \mathcal{N}_s \left( \mathbf{0}, C = \frac{1}{\eta^2} K \right)$ and consider the change of variables transformation: $\log(\theta_i) = \mu + \eta z_i$ where $\theta_i:=\theta(x_i)$, $z_i := z(x_i)$, $i=1,\ldots,s$. Then, the prior density of $\boldsymbol{\theta}$ is given by
\begin{eqnarray*}
\rho(\boldsymbol{\theta}|\mu , \eta, \ell) &=& \frac{(\eta^2 2\pi)^{-\frac{s}{2}} |C|^{-\frac{1}{2}}}{\theta_1 \,\theta_2 \cdots \theta_s} \exp{\left( -\frac{(\log{\boldsymbol{\theta}} - \boldsymbol{\mu})^{tr} C (\log{\boldsymbol{\theta}} - \boldsymbol{\mu})}{2 \eta^2} \right)} \\
&=& \frac{(2\pi)^{-\frac{s}{2}} |C|^{-\frac{1}{2}}}{\eta^s e^{s\mu+\eta(z_1+ \ldots +z_s)}} \exp{\left(-\frac{1}{2}\mathbf{z}^{tr} C^{-1} \mathbf{z} \right)} \\
&=& \frac{\rho(\mathbf{z}|\ell)}{\eta^s e^{s\mu+\eta(z_1+ \ldots +z_s)}}.
\end{eqnarray*}
The posterior density of the hyperparameters can be therefore written as
\begin{eqnarray*}
\rho(\mu , \ell, \eta |  \mathbf{Y_1}, \ldots, \mathbf{Y_N}) \propto \rho(\mu, \ell, \eta) \int_{\mathbf{Z}} \frac{\rho(\mathbf{z}|\ell)}{\eta^s e^{s\mu+\eta(z_1+ \ldots +z_s)}} \,|J|\, \rho( \mathbf{Y_1}, \ldots, \mathbf{Y_N}|\mu, \eta, \mathbf{z}) d\mathbf{z},
\end{eqnarray*}
where $J$ is the Jacobian matrix of the transformation.

By considering $\ell$ as a nuisance parameter, we obtain 
\begin{equation}
\label{eq:hpost}
\rho(\mu , \eta |  \mathbf{Y_1}, \ldots, \mathbf{Y_N}) \propto \rho(\mu, \eta) \int_{\ell} \rho(\ell) \int_{\mathbf{Z}} \rho(\mathbf{z}|\ell) \rho( \mathbf{Y_1}, \ldots, \mathbf{Y_N}|\mu, \eta, \mathbf{z}) d\mathbf{z} d\ell
\end{equation}
after $\ell$ is marginalized.

\begin{figure}[h!]
\centering
\includegraphics[scale=0.35]{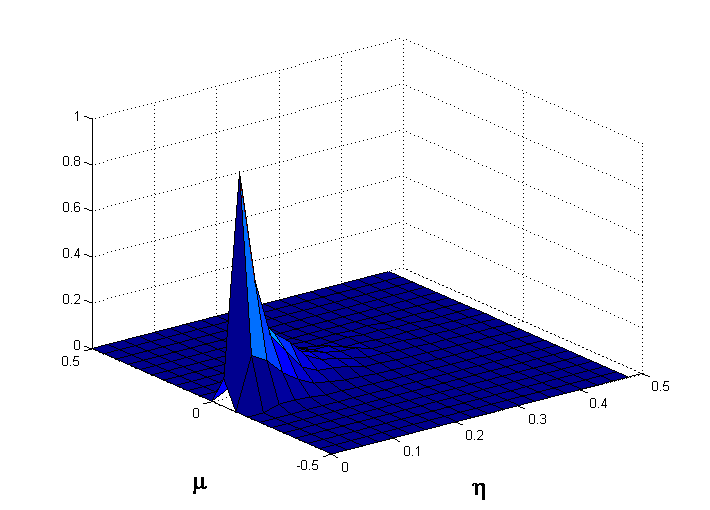}
\caption{Example 3: Non-normalized joint posterior density of the hyperparameters $\mu$ and $\eta$. The maximum a posteriori probability (MAP) estimate is $(-0.05,0.025)$.}
\label{pic:post2D1}
\end{figure}

\begin{figure}[h!]
\centering
\includegraphics[scale=0.35]{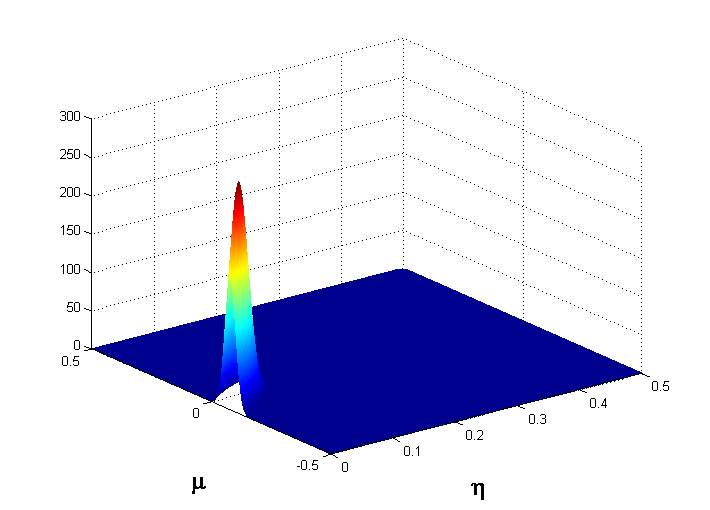}
\caption{Example 3: Laplace's approximation of the posterior density of the hyperparameters $\mu$ and $\eta$.}
\label{pic:laplace2D1}
\end{figure}

To evaluate the posterior distribution of the hyperparametrs, we need to compute the $s+1$ dimensional integral in formula \eqref{eq:hpost}. Alternatively, Monte Carlo method can be used to approximate these integrations. First, we sample $\ell$ from its prior distribution, $\rho(\ell)$. Then, given $\ell$ we can sample $\mathbf{z}$ and evaluate the joint likelihood function, $\rho( \mathbf{Y_1}, \ldots, \mathbf{Y_N}|\mu, \eta, \mathbf{z})$, for any pair $(\mu,\eta)$. Therefore, we approximate the non-normalized posterior distribution using a double sum as follows:
\begin{eqnarray*}
\int_{\ell} \rho(\ell) \int_{\mathbf{Z}} \rho(\mathbf{z}|\ell) \rho( \mathbf{Y_1}, \ldots, \mathbf{Y_N}|\mu, \eta, \mathbf{z}) d\mathbf{z} d\ell
&\approx& \frac{1}{M_\ell}\sum_{i=1}^{M_\ell} \int_{\mathbf{Z}} \rho(\mathbf{z}|\ell_{i}) \rho( \mathbf{Y_1}, \ldots, \mathbf{Y_N}|\mu, \eta, \mathbf{z}) d\mathbf{z}\\ 
&\approx&  \frac{1}{M_\ell}  \frac{1}{M_z} \sum_{i=1}^{M_\ell} \sum_{j=1}^{M_z} \rho( \mathbf{Y_1}, \ldots, \mathbf{Y_N}|\mu, \eta, \mathbf{z}_{j}),
\end{eqnarray*}
where $\mathbf{z}_{j} \sim \rho(\mathbf{z}|\ell_{i})$.

Figure \ref{pic:post2D1} shows that the non-normalized posterior of the hyperparameters $\mu$ and $\eta$ has a unique mode at $(-0.05,0.025)$. We use then Laplace's method to obtain a Gaussian posterior, using the synthetic dataset A, as shown in Figure \ref{pic:laplace2D1}.

A new dataset is now introduced to test our method when $\theta$ depends on $x$.

\textbf{Dataset B}\\
To analyze the performance of our inferential technique in the case where the thermal diffusivity parameter depends on the space variable, $x$, we consider another synthetic dataset (hereafter named dataset B) that is generated similarly to the dataset A, except for the fact that $\theta(x)$ is sampled randomly from the new prior \eqref{priormulti} where $\mu=0, \eta = 0.1$ and $\ell = 5$.

Again, we approximate the posterior distribution of the hyperparameters $\mu$ and $\eta$ using Laplace's approximation given the following priors for the hyperparameters $\mu, \eta$ and $\ell$:
\[ \mu \sim \mathcal{N} \left( 0, 0.25 \right), \quad
\eta \sim \textrm{half-Cauchy}\left( 0.5 \right), \quad
\ell \sim U \left(4, 6 \right), \]
where we assume broad priors for $\mu$ and $\eta$ with a more informative unifrom prior for $\ell$.

\begin{figure}[h!]
\centering
\includegraphics[scale=0.35]{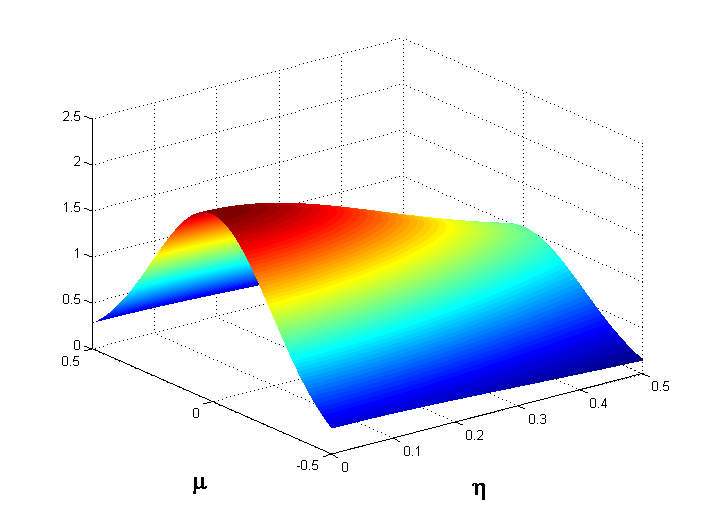}
\caption{Example 3: Joint prior density for the hyperparameters $\mu$ and $\eta$, with $\mu \sim \mathcal{N} \left( 0, 0.25 \right)$ and $\eta \sim \textrm{half-Cauchy}\left( 0.5 \right)$.}
\label{pic:prior2D2}
\end{figure}

\begin{figure}[h!]
\centering
\includegraphics[scale=0.35]{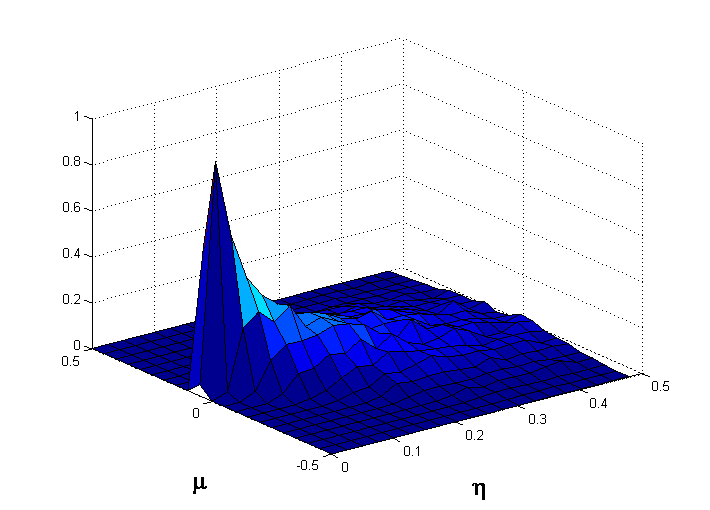}
\caption{Example 3: Non-normalized joint posterior density of the hyperparameters $\mu$ and $\eta$. The maximum a posteriori probability (MAP) estimate is $(0.05,0.025)$.}
\label{pic:post2D2}
\end{figure}

\begin{figure}[h!]
\centering
\includegraphics[scale=0.35]{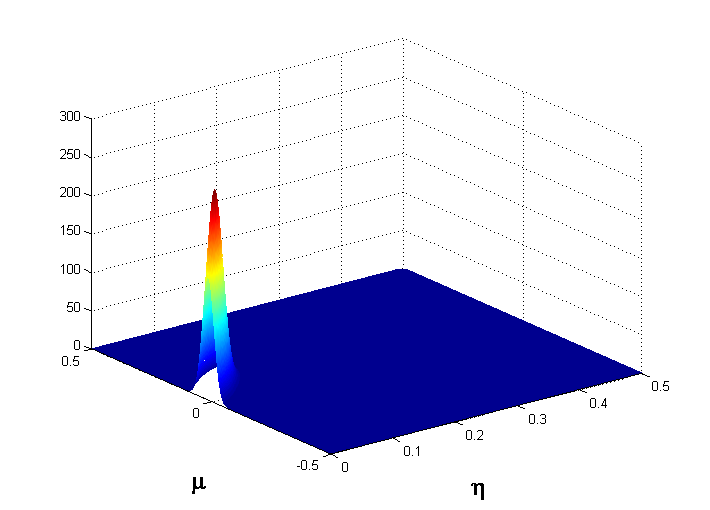}
\caption{Example 3: Laplace's approximation for the posterior density of the hyperparameters $\mu$ and $\eta$.}
\label{pic:laplace2D2}
\end{figure}

From Figure \ref{pic:post2D2}, we find that the maximum a posteriori probability (MAP) estimate is $(0.05,0.025)$ and Laplace's approximation can be used.
By comparing the prior and posterior densities for $\mu$ and $\eta$ in Figures \ref{pic:prior2D2} and \ref{pic:laplace2D2}, we can say that the experiment is informative since the posterior concentrates around $(0.05,0.025)$ which is close to the true value.

\section{Conclusion}
In this work, we developed a general Bayesian approach for one-dimensional linear parabolic partial differential equations with noisy boundary conditions. First, we derived the joint likelihood of the thermal diffusivity $\theta$ and the boundary parameters. Second, we approximated the solution of the forward problem, by showing that such solution can be written as a linear function of the boundary conditions. After that, we marginalized out the boundary parameters, under the assumptions that they are well approximated by piecewise linear functions and that they are independent Gaussian random vectors. This approach can be generalized to any well-posed linear partial differential equation.

On the implementation side, we computed the log-posterior of the thermal diffusivity in different cases. Besides, we used the Laplace approximation to obtain a Gaussian posterior. In the first example, we used directly the joint likelihood of the thermal diffusivity $\theta$ and the boundary parameters, assuming that the boundary conditions were known. In the second example, we used the marginalized likelihood of $\theta$, assuming that $\theta$ is a lognormal random variable and, as in the previous example, we obtained an approximated Gaussian posterior distribution, showing that the unknown value of the thermal diffusivity is recovered almost exactly. Moreover we explored two important advantages of using the Bayesian approach, by providing the estimation of the expected information gain for different experimental setups and the predictive posterior distribution of the temperature. We noticed that the temperature measurements from the middle thermocouple at the second time interval are, in general, the most informative measurements.  Finally, we considered the case where $\theta$ is a lognormal random field with squared exponential covariance function. In this case, we obtained the joint posterior distribution for the covariance hyperparameters by applying hierarchical Bayesian techniques.

\bibliographystyle{ba}
\bibliography{reference}

\begin{acknowledgement}
Part of this work was carried out while F. Ruggeri and M. Scavino were Visiting Professors at KAUST. Z. Sawlan, M. Scavino and R. Tempone are members of the KAUST SRI Center for Uncertainty Quantification in Computational Science and Engineering.
\end{acknowledgement}

\appendix
\section*{Appendix A}
\label{appA}
Proof of Theorem \ref{th:linear}:
%
%

First, let us introduce the vectors
\begin{eqnarray*}
\mathop{F_{L,1}}\limits_{(I-1) \times 1} &=& \left[ - \left(\int_{x_L}^{x_R} \frac{x_R - x}{x_R - x_L}\phi_{j} dx \right)_{j} \right], \\
\mathop{F_{L,2}}\limits_{(I-1) \times 1} &=& \left[ \left(\int_{x_L}^{x_R} \left( \frac{x_R - x}{x_R - x_L} - \Delta t L_{\boldsymbol{\theta}} \frac{x_R - x}{x_R - x_L} \right)\phi_{j} dx \right)_{j} \right], \\
\mathop{F_{R,1}}\limits_{(I-1) \times 1} &=& \left[ - \left(\int_{x_L}^{x_R} \frac{x-x_L}{x_R - x_L}\phi_{j} dx \right)_{j} \right], \\
\mathop{F_{R,2}}\limits_{(I-1) \times 1} &=& \left[ \left(\int_{x_L}^{x_R} \left( \frac{x-x_L}{x_R - x_L} - \Delta t L_{\boldsymbol{\theta}} \frac{x-x_L}{x_R - x_L} \right)\phi_{j} dx \right)_{j} \right],
\end{eqnarray*}

and the matrices
\begin{eqnarray*}
\mathop{\mathbf{B}}\limits_{(I-1) \times (I-1)} &=& \left( M + \Delta t S_{\boldsymbol{\theta}} \right)^{-1} M,\\
\mathop{\mathbf{F}_{L,n}}\limits_{(I-1) \times N+1} &=&  \left(
\begin{array}{cccc}
\mathop{\mathbf{0}}\limits_{(I-1) \times (n-1)} &
\mathop{F_{L,1}}\limits_{(I-1) \times 1} &
\mathop{F_{L,2}}\limits_{(I-1) \times 1} &
\mathop{\mathbf{0}}\limits_{(I-1) \times (N-n)}
\end{array} \right), \: n = 2,\ldots , N-1, \\
\mathop{\mathbf{F}_{R,n}}\limits_{(I-1) \times N+1} &=&  \left(
\begin{array}{cccc}
\mathop{\mathbf{0}}\limits_{(I-1) \times (n-1)} &
\mathop{F_{R,1}}\limits_{(I-1) \times 1} &
\mathop{F_{R,2}}\limits_{(I-1) \times 1} &
\mathop{\mathbf{0}}\limits_{(I-1) \times (N-n)}
\end{array} \right), \: n = 2,\ldots , N-1, \\
\mathop{\mathbf{F}_{L,1}}\limits_{(I-1) \times N+1} &=& \left(
\begin{array}{ccc}
\mathop{F_{L,1}}\limits_{(I-1) \times 1} &
\mathop{F_{L,2}}\limits_{(I-1) \times 1} &
\mathop{\mathbf{0}}\limits_{(I-1) \times (N-1)}
\end{array} \right), \\
\mathop{\mathbf{F}_{R,1}}\limits_{(I-1) \times N+1} &=& \left(
\begin{array}{ccc}
\mathop{F_{R,1}}\limits_{(I-1) \times 1} &
\mathop{F_{R,2}}\limits_{(I-1) \times 1} &
\mathop{\mathbf{0}}\limits_{(I-1) \times (N-1)}
\end{array} \right), \\
\mathop{\mathbf{F}_{L,N}}\limits_{(I-1) \times N+1} &=& \left(
\begin{array}{ccc}
\mathop{\mathbf{0}}\limits_{(I-1) \times (N-1)} &
\mathop{F_{L,1}}\limits_{(I-1) \times 1} &
\mathop{F_{L,2}}\limits_{(I-1) \times 1} 
\end{array} \right), \\
\mathop{\mathbf{F}_{R,N}}\limits_{(I-1) \times N+1} &=& \left(
\begin{array}{ccc}
\mathop{\mathbf{0}}\limits_{(I-1) \times (N-1)} &
\mathop{F_{R,1}}\limits_{(I-1) \times 1} &
\mathop{F_{R,2}}\limits_{(I-1) \times 1} 
\end{array} \right).
\end{eqnarray*}

Then, the solution of \eqref{eq:hom} is given by:
\begin{equation*}
\mathbf{u}_{n+1} = B \mathbf{u}_{n} + \left( M + \Delta t S_{\boldsymbol{\theta}} \right)^{-1} \left( \mathbf{F}_{L,n}\mathbf{T}_{L} + \mathbf{F}_{R,n}\mathbf{T}_{R} \right).
\end{equation*}

Applying recursively the previous relation we derive the discrete representation (Duhamel's formula)
\begin{equation*}
\mathbf{u}_{n} = \mathbf{B}^n \mathbf{u}_0 + \sum_{k=1}^{n} \mathbf{B}^{n-k} \left( M + \Delta t S_{\boldsymbol{\theta}} \right)^{-1} \left( \mathbf{F}_{L,k}\mathbf{T}_{L} + \mathbf{F}_{R,k}\mathbf{T}_{R} \right).
\end{equation*}

Now we can build the matrices $A_n(\boldsymbol{\theta}),\tilde{A}_{L,n}(\boldsymbol{\theta})$ and $\tilde{A}_{R,n}(\boldsymbol{\theta}) \,, n=1,\ldots,N\,,$ introduced in the expression \eqref{eq:linear}, to recover the solution of the problem \eqref{eq:hom} as a linear function of the initial-boundary conditions, namely:
\begin{eqnarray*}
A_{n}(\boldsymbol{\theta}) &=& \mathbf{B}^{n},\\
\tilde{A}_{L,n}(\boldsymbol{\theta}) &=& \sum_{k=1}^{n} \mathbf{B}^{n-k} \left( M + \Delta t S_{\boldsymbol{\theta}} \right)^{-1} \mathbf{F}_{L,k},\:\textrm{and}\\
\tilde{A}_{R,n}(\boldsymbol{\theta}) &=& \sum_{k=1}^{n} \mathbf{B}^{n-k} \left( M + \Delta t S_{\boldsymbol{\theta}} \right)^{-1} \mathbf{F}_{R,k}. \qed
\end{eqnarray*}

\section*{Appendix B}
\label{appB}
Proof of Theorem \ref{th:linear2}:

From \eqref{eq:lemma} and \eqref{eq:linear}, we can write:
\begin{equation*}
\mathbf{T}_{n} = \mathbf{B}^{n} \mathbf{u}_0 + \tilde{A}_{L,n}(\boldsymbol{\theta}) {\mathbf{T}}_L + 
\tilde{A}_{R,n}(\boldsymbol{\theta}) {\mathbf{T}}_{R} - T_{L,n}F_{L,1} -T_{R,n}F_{R,1},
\end{equation*}
Now, define $A_{L,n}(\boldsymbol{\theta})$ and $A_{R,n}(\boldsymbol{\theta})$ by:
\begin{eqnarray*}
A_{L,n}(\boldsymbol{\theta}) &=& \tilde{A}_{L,n}(\boldsymbol{\theta}) + \left(
\begin{array}{cccc}
\mathop{\mathbf{B}^{n} F_{L,1}}\limits_{(I-1) \times 1} &
\mathop{\mathbf{0}}\limits_{(I-1) \times (n-1)} &
\mathop{-F_{L,1}}\limits_{(I-1) \times 1} &
\mathop{\mathbf{0}}\limits_{(I-1) \times (N-n)}
\end{array} \right),\\
A_{R,n}(\boldsymbol{\theta}) &=& \tilde{A}_{R,n}(\boldsymbol{\theta}) + \left(
\begin{array}{cccc}
\mathop{\mathbf{B}^{n} F_{R,1}}\limits_{(I-1) \times 1} &
\mathop{\mathbf{0}}\limits_{(I-1) \times (n-1)} &
\mathop{-F_{R,1}}\limits_{(I-1) \times 1} &
\mathop{\mathbf{0}}\limits_{(I-1) \times (N-n)}
\end{array} \right).
\end{eqnarray*}
Therefore, we obtain equation \eqref{eq:mainsol}.
\qed

\section*{Appendix C}
\label{appC}
Proof of Remark \ref{re:ex2}:

Consider the following backward Euler discretization of the local problem \eqref{locprob} in the interval time, $(t_n, t_{n+1}) = (n \Delta t, (n+1) \Delta t)$:
\begin{equation}
 \left\{
\begin{array}{rl}
\frac{1}{\Delta t} (T_{i, n+1} - & T_{i,n}) - \frac{\theta}{\Delta x ^2} (T_{i+1, n+1} -2 T_{i, n+1} + T_{i-1, n+1}) = 0,  \:\:i = 2,\ldots,I \\
T_{L,n} =  & T_{L}(n \Delta t)\,, \\
T_{R,n} =  & T_{R}(n \Delta t)\,.
\end{array} \right.
\label{backeuler}
\end{equation}
To write the discretization \eqref{backeuler} in a matrix form, let us introduce the vectors
$\mathop{\mathbf{T}_{n}}\limits_{(I-1) \times 1} = (T_{2,n}, \ldots, T_{I,n})^{tr}$, $n= 1, \ldots, N$, and the matrix
\begin{equation*}
\mathop{\mathbf{A}}\limits_{(I-1) \times (I-1)} = \left(
\begin{array}{cccccc}
-2 & 1 & 0 & 0 &  \ldots & 0 \\
1 & -2 & 1 & 0 & \ldots & 0 \\
0 & 1 & -2 & 1 & \ldots & 0 \\
\vdots & \ddots & \ddots & \ddots & \ddots & \vdots \\
0 & \ldots & 0 & 1 & -2 & 1\\
0 & 0 & 0 & \ldots & 1 & -2
\end{array} \right) \,.
\end{equation*}

In this way, we may write
\begin{equation} \frac{1}{\Delta t} (\mathbf{T}_{n+1} - \mathbf{T}_{n}) - \frac{\theta}{\Delta x^2} \mathbf{A} \mathbf{T}_{n+1}
= \frac{\theta}{\Delta x^2} (T_{L,n+1} \mathbf{v} + T_{R,n+1} \mathbf{w}) \,, \label{eulermat} \end{equation}
where $\mathop{\mathbf{v}}\limits_{(I-1) \times 1} = (1, 0, \ldots, 0)^{tr}$ and $\mathop{\mathbf{w}}\limits_{(I-1) \times 1} = (0, \ldots, 0, 1)^{tr}\,.$ \\
The expression \eqref{eulermat} is equal to
\[ (I_{I-1} - \theta \frac{\Delta t}{\Delta x^2} \mathbf{A} ) \mathbf{T}_{n+1} = \mathbf{T}_{n} + \theta \frac{\Delta t}{\Delta x^2}
(T_{L,n+1} \mathbf{v} + T_{R,n+1} \mathbf{w})  \]
and letting $\frac{\Delta t}{\Delta x^2} = \lambda$ and $\mathbf{B} = (I_{I-1} - \theta \lambda \mathbf{A})^{-1}$, we obtain
\[ \mathbf{T}_{n+1} = \mathbf{B} \mathbf{T}_{n} + \theta \lambda ( T_{L,n+1} \mathbf{B} \mathbf{v} + T_{R,n+1} \mathbf{B} \mathbf{w} ) \,.   \]

Applying recursively the previous relation, we derive
\[ \mathbf{T}_{n} = \mathbf{B}^n \mathbf{T}_0 + \theta \lambda \sum_{k=1}^{n} T_{L,k} \mathbf{B}^{n-k+1} \mathbf{v} + \theta \lambda \sum_{k=1}^{n} T_{R,k} \mathbf{B}^{n-k+1} \mathbf{w} \,, \]
whose compact matrix form is
\[ \mathbf{T}_{n} = \mathbf{B}^n \mathbf{T}_0 + \mathbf{C}_n \tilde{\mathbf{T}}_{L} + \mathbf{D}_n \tilde{\mathbf{T}}_{R} \,,\]
where
\begin{itemize}
\item $\mathop{\tilde{\mathbf{T}}_L}\limits_{n \times 1} = (T_{L,1}, \ldots, T_{L,n})^{tr} = (T_L(\Delta t), \ldots, T_L(n \Delta t))^{tr}\,,$
\item $\mathop{\tilde{\mathbf{T}}_{R}}\limits_{n \times 1} = (T_{R,1}, \ldots, T_{R,n})^{tr} = (T_R(\Delta t), \ldots, TR( n \Delta t))^{tr}\,,$
\item the matrix $\mathop{\mathbf{C}_n}\limits_{(I-1) \times n}$ has column vectors $ \mathbf{c}_k = \theta \lambda \mathbf{A}^{n-k+1} \mathbf{v} \,, \: k=1,\ldots,n \,,$
\item the matrix $\mathop{\mathbf{D}_n}\limits_{(I-1) \times n}$ has column vectors $ \mathbf{d}_k = \theta \lambda \mathbf{A}^{n-k+1} \mathbf{w} \,, \: k=1,\ldots,n \,.$
\end{itemize}
Now, we can build the matrices $ \mathop{A_{L,n}(\theta)}\limits_{(I-1) \times N}$ and $ \mathop{A_{R,n}(\theta)}\limits_{(I-1) \times N} \,, n=1,\ldots,N\,,$ introduced in the expression \eqref{eq:res}, to recover the solution of the problem \eqref{locprob} for each interval time, $(t_{n-1}, t_n)\,,n=1, \ldots, N\,,$ as a linear function of the initial-boundary conditions:
\begin{equation*}
\mathop{A_{L,n}(\theta)}\limits_{(I-1) \times N} =  \left(
\begin{array}{cc}
\mathop{\mathbf{C}_n}\limits_{(I-1) \times n} & \mathop{\mathbf{0}}\limits_{(I-1) \times (N-n)}
\end{array} \right) \,,
\end{equation*}
\begin{equation*}
\mathop{A_{R,n}(\theta)}\limits_{(I-1) \times N} =  \left(
\begin{array}{cc}
\mathop{\mathbf{C}_n}\limits_{(I-1) \times n} & \mathop{\mathbf{0}}\limits_{(I-1) \times (N-n)}
\end{array} \right) \,.
\end{equation*} \qed

\end{document}